\PassOptionsToPackage{dvipsnames}{xcolor}
\documentclass[sigconf,nonacm]{aamas}
\usepackage{balance}
\usepackage{microtype,amsthm,mathtools,nicefrac,xspace}
\usepackage{xcolor}

\usepackage{tikz,booktabs,array,graphicx,placeins}
\usetikzlibrary{arrows.meta,positioning}
\usepackage{xurl}
\usepackage{algorithm}
\usepackage{algpseudocode}
\hypersetup{colorlinks=true,linkcolor=MidnightBlue,citecolor=MidnightBlue,urlcolor=MidnightBlue,
pdftitle={Parallel AI Agent Collaboration as Graph-Structured Distributed Block Optimization},
pdfauthor={Hanmin Li, Yury Demidovich, Peter Richt\'arik},pdfsubject={},pdfkeywords={multi-agent systems, agentic AI, distributed block optimization, fractional graph coloring}}

\newcommand{\eqdef}{\coloneqq}
\newcommand{\Exp}[1]{\mathbb E\left[#1\right]}
\newcommand{\ExpCond}[2]{\mathbb E\left[#1\,\middle|\,#2\right]}

\newcommand{\cX}{\mathcal X}
\newcommand{\cH}{\mathcal H}
\newcommand{\cM}{\mathcal M}

\newcommand{\cI}{\mathcal I}
\DeclareMathOperator*{\argmax}{arg\,max}
\theoremstyle{definition}
\newtheorem{assumption}{Assumption}
\theoremstyle{plain}
\newtheorem{theorem}{Theorem}

\newtheorem{proposition}[theorem]{Proposition}
\newtheorem{corollary}[theorem]{Corollary}

\makeatletter
\gdef\@copyrightpermission{
  \begin{minipage}{0.2\columnwidth}
   \href{https://creativecommons.org/licenses/by/4.0/}{\includegraphics[width=0.90\textwidth]{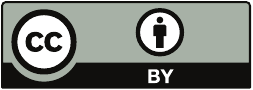}}
  \end{minipage}\hfill
  \begin{minipage}{0.8\columnwidth}
   \href{https://creativecommons.org/licenses/by/4.0/}{This work is licensed under a Creative Commons Attribution International 4.0 License.}
  \end{minipage}
  \vspace{5pt}
}
\makeatother
\setcopyright{none}
\acmConference[Preprint]{arXiv working draft}{October 2026}{}{}
\copyrightyear{2026}
\acmYear{2026}
\acmDOI{}
\acmISBN{}
\acmSubmissionID{}
\submissionType{}

\title[Parallel AI Agent Collaboration]{Parallel AI Agent Collaboration as Graph-Structured Distributed Block Optimization}
\makeatletter
\patchcmd{\@typeset@author@bx}
  {\@authorfont\@currentauthors\par}
  {\@authorfont\begingroup
   \def\par{\ifhmode\qquad\fi}\@currentauthors\endgroup\par}
  {}{\PackageError{author-layout}{Unable to set horizontal author layout}{}}
\makeatother
\author{Hanmin Li}
\author{Yury Demidovich}
\author{Peter Richt\'arik}
\affiliation{
  \institution{King Abdullah University of Science and Technology}
  \city{Thuwal}
  \country{Saudi Arabia}}

\begin{abstract}
Parallel AI agents can each produce useful local edits yet degrade a shared solution when their edits are combined. We model this coordination problem as discrete block optimization with an exactly evaluable objective and a known dependency graph. The graph certifies compatible sets of edits through its independent sets, and evaluating proposals against the same current solution yields verified gains that add exactly within any such set. A coordinator can therefore construct a safe compatible candidate while additionally evaluating interacting edits to capture favorable interactions. We prove that, using an optimal fractional coloring, the safe merger preserves a $1/\chi_f(G)$ fraction of the total verified local improvement, where $\chi_f(G)$ is the fractional chromatic number of the dependency graph. We further show that no merger restricted to selecting a subset of the same proposals can guarantee a larger fraction in general. Under conditional local-improvement guarantees and a block error-bound condition, repeated collaboration converges geometrically in the objective gap, separating the effects of task structure, agent allocation, and computation cost. The resulting rate gives a linear program for agent assignment and an explicit time--cost trade-off. We complement the analysis with experiments on solution quality and the cost of parallel coordination. Code-repair experiments show that verified merging resolves conflicts between individually useful edits that cause naive merging to fail. On larger conflict-free tasks, parallel agents reduce modeled completion time at higher input-token-equivalent cost. On three public road networks, fractional integration retains 922 units of verified graph-cut improvement across 90 paired runs. For route-search agents sharing exclusive access to vehicle routes, fractional scheduling reduces measured wall time by a median of 15\% across twelve paired runs, with route-cost changes varying across runs.
\end{abstract}


\begin{document}
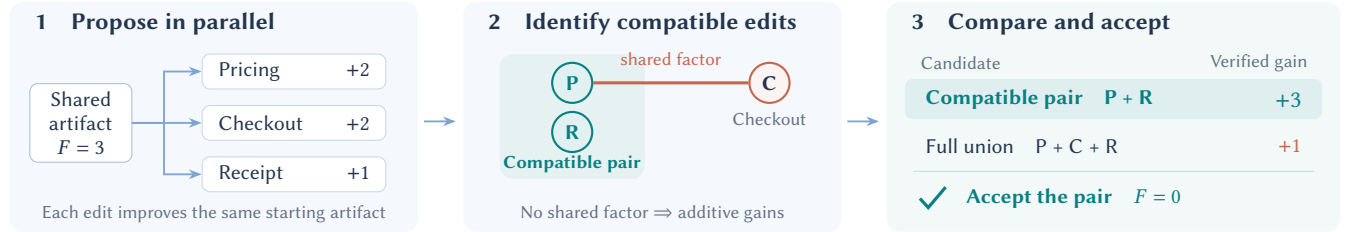
\begin{teaserfigure}
\centering
\begingroup
\definecolor{teaserInk}{HTML}{23354D}
\definecolor{teaserBlue}{HTML}{3975BA}
\definecolor{teaserTeal}{HTML}{087F82}
\definecolor{teaserCoral}{HTML}{C4674E}
\definecolor{teaserMute}{HTML}{637184}
\begin{tikzpicture}[x=1cm,y=1cm,
  every node/.style={font=\sffamily\fontsize{8}{9}\selectfont,text=teaserInk},
  heading/.style={font=\sffamily\bfseries\fontsize{9}{10}\selectfont,anchor=west},
  card/.style={rounded corners=3pt,draw=teaserBlue!28,fill=white,
    minimum width=2.43cm,minimum height=.46cm,text width=2.0cm,align=left},
  vertex/.style={circle,draw=teaserBlue,line width=.8pt,fill=white,
    minimum size=.53cm,inner sep=0pt,font=\sffamily\bfseries\small},
  selected/.style={vertex,draw=teaserTeal,fill=teaserTeal!12,text=teaserTeal},
  flow/.style={-{Stealth[length=1.7mm]},draw=teaserBlue!65,line width=.7pt}]
\path[use as bounding box] (0,0) rectangle (17.6,3.15);

\fill[teaserBlue!5,rounded corners=5pt] (0,.08) rectangle (5.4,3.12);
\fill[teaserBlue!5,rounded corners=5pt] (6.0,.08) rectangle (11.0,3.12);
\fill[teaserTeal!5,rounded corners=5pt] (11.6,.08) rectangle (17.6,3.12);
\node[heading] at (.22,2.83) {1\quad Propose in parallel};
\node[heading] at (6.22,2.83) {2\quad Identify compatible edits};
\node[heading] at (11.82,2.83) {3\quad Compare and accept};

\node[draw=teaserBlue!45,rounded corners=3pt,fill=white,
      minimum width=1.35cm,minimum height=1.0cm,align=center] (base) at (.94,1.52)
      {Shared\\artifact\\$F=3$};
\node[card] (p) at (3.76,2.20) {Pricing\hfill $+2$};
\node[card] (c) at (3.76,1.52) {Checkout\hfill $+2$};
\node[card] (r) at (3.76,.84) {Receipt\hfill $+1$};
\foreach \target in {p,c,r}{\draw[flow] (base.east) -- (2.04,1.52) |- (\target.west);}
\node[text=teaserMute,font=\sffamily\fontsize{7}{8}\selectfont] at (2.7,.32)
      {Each edit improves the same starting artifact};

\fill[teaserTeal!10,rounded corners=5pt] (6.48,.78) rectangle (8.40,2.43);
\node[selected] (gp) at (7.44,2.05) {P};
\node[selected] (gr) at (7.44,1.40) {R};
\node[vertex,draw=teaserCoral,fill=teaserCoral!7] (gc) at (10.05,2.05) {C};
\draw[teaserCoral,line width=1.15pt] (gp.east) --
      node[above=3pt,font=\sffamily\fontsize{7}{8}\selectfont,text=teaserCoral] {shared factor} (gc.west);
\node[text=teaserTeal,font=\sffamily\bfseries\fontsize{7}{8}\selectfont] at (7.44,.98) {Compatible pair};
\node[text=teaserMute,font=\sffamily\fontsize{7}{8}\selectfont] at (10.05,1.57) {Checkout};
\node[text=teaserMute,font=\sffamily\fontsize{7}{8}\selectfont] at (8.5,.32)
      {No shared factor $\Rightarrow$ additive gains};

\node[anchor=west,font=\sffamily\fontsize{7}{8}\selectfont,text=teaserMute] at (11.94,2.32) {Candidate};
\node[anchor=east,font=\sffamily\fontsize{7}{8}\selectfont,text=teaserMute] at (17.28,2.32) {Verified gain};
\fill[teaserTeal!13,rounded corners=3pt] (11.84,1.57) rectangle (17.36,2.10);
\node[anchor=west,text=teaserTeal,font=\sffamily\bfseries\fontsize{8}{9}\selectfont] at (12.0,1.83) {Compatible pair\quad P + R};
\node[anchor=east,text=teaserTeal,font=\sffamily\bfseries\fontsize{9}{10}\selectfont] at (17.2,1.83) {$+3$};
\node[anchor=west] at (12.0,1.22) {Full union\quad P + C + R};
\node[anchor=east,text=teaserCoral] at (17.2,1.22) {$+1$};
\draw[teaserTeal!25] (11.95,.88) -- (17.25,.88);
\draw[teaserTeal,line width=1.2pt] (12.04,.52) -- (12.15,.40) -- (12.38,.67);
\node[anchor=west,text=teaserTeal,font=\sffamily\bfseries\fontsize{8}{9}\selectfont] at (12.54,.52)
      {Accept the pair\quad $F=0$};

\draw[flow] (5.48,1.54) -- (5.9,1.54);
\draw[flow] (11.08,1.54) -- (11.5,1.54);
\end{tikzpicture}
\endgroup

\caption{Verified integration of parallel edits. In the running example, three individually improving edits give a gain of only one when merged together. The compatible pricing--receipt pair gains three and reaches the optimum.}
\Description{Three agents propose edits to pricing, checkout, and receipt, with individual gains two, two, and one. The dependency graph connects pricing and checkout. The compatible pricing and receipt pair gains three, compared with one for the full union, and is selected.}
\label{fig:teaser}
\end{teaserfigure}
\maketitle
\hypersetup{pdfauthor={Hanmin Li, Yury Demidovich, Peter Richt\'arik}}
\pdfstringdef\arxivpdfauthors{Hanmin Li, Yury Demidovich, Peter Richt\'arik}
\pdfinfo{/Author(\arxivpdfauthors)}

\section{Introduction}

Many tasks handled by AI agents can be decomposed into components that are worked on in parallel.
This can reduce latency and distribute work across components, but it also creates a coordination problem.
An edit that improves one component in isolation may interact with edits to other components once they are combined.
Program repair provides a concrete example: agents may modify different modules in parallel, yet their changes can still interfere through shared functionality or interfaces.
The central question is therefore not only how to generate useful local edits, but also how to decide which of them can be safely integrated into the shared solution.

For example, suppose a program must apply a discount exactly once, and either the pricing module or the checkout module can implement the required transformation.
Starting from a version that does not apply the discount, an edit to either module improves the program on its own.
If both edits are applied, however, the discount is applied twice.
Thus, individually useful edits to different modules need not be compatible.

The example exposes a broader coordination problem in parallel agent work.
A proposed edit can be evaluated in isolation against the current artifact, yet its value may change when other edits are applied alongside it.
We therefore separate proposal generation from proposal integration.
Agents work concurrently from the same artifact on distinct blocks, while a coordinator uses an evaluator and dependency information to decide which returned edits to combine.
Our goal is to determine how much of the verified local improvement can be guaranteed to be retained after integration.

We call this problem \emph{verified proposal integration}.
The coordinator observes black-box block proposals, their verified singleton gains against a common base artifact, and a dependency certificate identifying subsets whose gains add exactly.
It need not know how the proposals were generated or model their joint effects beyond this certificate.
We ask how much of the verified local improvement a coordinator can always preserve using only this information.
Our main result is a tight worst-case bound for integrating fixed proposals. We extend the analysis to convergence and agent allocation under additional assumptions. Experiments examine the integration mechanism and its cost.

\paragraph{Contributions.}
\begin{itemize}
\item We formulate \emph{verified proposal integration} for black-box agent edits and characterize its exact worst-case retention factor. 
Given verified singleton gains and a dependency certificate, a coordinator can always preserve a $1/\chi_f(G)$ fraction of the total verified local improvement. 
This guarantee is tight: in the worst case, no coordinator that forms its update by selecting a subset of the submitted proposals can guarantee a larger fraction (Sections~\ref{sec:model}--\ref{sec:merge}).

\item Under explicit assumptions on agent capability, assignment coverage, and task structure, we establish geometric convergence of repeated verified integration.
The convergence rate separates the effects of these three factors and yields a principled formulation for agent assignment and explicit time--cost trade-offs (Sections~\ref{sec:theory}--\ref{sec:allocation}).

\item We test proposal integration and its costs with language-model agents and use controlled examples to clarify the guarantees.
The experiments examine conflict resolution among individually useful edits, performance on held-out multi-round tasks, and the time--cost trade-off of parallel execution. A road-network diagnostic evaluates exact-factor merging on naturally occurring block interactions and compares fractional integration with fixed and adaptive ordinary coloring (Section~\ref{sec:experiments}).
\end{itemize}

\paragraph{Related work.}
Our formulation builds on ideas from parallel and distributed optimization, where multiple components are updated concurrently while accounting for their interactions.
In particular, it is closely related to parallel block coordinate descent and arbitrary sampling, which study how subsets of coordinates can be updated simultaneously \citep{richtarik2016,qu2016}.
Within this line of work, \citet{bradley2011parallel} analyze interference among parallel coordinate updates using a spectral measure.
A complementary perspective appears in GenCD, which separates coordinate descent into selection, proposal, acceptance, and update stages and includes a coloring-based variant \citep{scherrer2012scalingcoordinatedescentalgorithms}.
Our setting differs in two important ways.
First, proposals may be arbitrary discrete edits rather than gradient-based coordinate updates.
Second, the coordinator observes verified singleton gains and a dependency certificate without requiring a smooth model of their joint effects.

Distributed constraint optimization (DCOP) coordinates agents that control different variables to optimize a shared objective expressed through local functions \citep{modi2005adopt,fioretto2018distributed,troullinos2025conditional}.
This is closely related in spirit to our use of localized dependencies, but the coordination problem is different.
In DCOP, agents coordinate the choice of variable assignments themselves.
In our setting, agents have already generated candidate edits, and the coordinator decides which submitted proposals to combine using their verified singleton gains and a dependency certificate.
Our focus is therefore on verified integration of independently generated proposals rather than distributed search over the solution space.

Language-model agents have increasingly been applied to software development and repair.
MetaGPT and ChatDev organize software development through structured roles and multi-agent collaboration \citep{hong2024,qian2024}.
MAGIS and MASAI decompose repository-level software tasks across specialized agents \citep{NEURIPS2024_5d1f0213,arora2024masai}.
RepairAgent instead focuses on autonomous program repair through tool use and iterative validation \citep{bouzenia2025repairagent}, while a recent survey reviews the broader use of multi-agent systems across the software-development lifecycle \citep{he2025llm}.
These works focus primarily on how software agents generate, coordinate, and refine solutions.
We instead isolate the integration step, where independently generated edits must be reconciled into a single shared artifact.

A separate line of work studies how collaboration structures among agents should be designed or adapted.
GPTSwarm represents agent systems as optimizable computational graphs \citep{zhuge2024}, while DyLAN dynamically selects agents and communication structures \citep{liu2024}.
Hyperagent uses hypergraphs for the topology of multi-agent communication \citep{zhang2026hyperagent}.
These methods focus on structuring how agents interact during problem solving.
By contrast, our dependency graph represents interactions among components of the shared artifact rather than communication links among agents.
It identifies subsets of proposed edits whose verified gains add exactly and thereby supports their certified integration.

\section{Problem Formulation}\label{sec:model}
\paragraph{Blocks and objective.}
We represent the current shared artifact as $x=(x_b)_{b=1}^m$, where each block $x_b$ corresponds to one component of the solution, such as a program module.
For $b\in[m]\eqdef\{1,\ldots,m\}$, let $\cX_b$ be the nonempty finite set of admissible values for block $b$.
The resulting search space is $\cX=\prod_{b=1}^m\cX_b$.

We evaluate an artifact using a known loss function $F:\cX\to\mathbb R$, where smaller values are better.
For example, in program repair, $F$ may count failed checks or combine penalties derived from an explicit evaluation rubric.
For $y\in\cX_b$, we write $(x_{-b},y)$ for the artifact obtained by replacing $x_b$ with $y$ while leaving all other blocks unchanged.
We define
\begin{equation}\label{eq:task}
\begin{gathered}
F^\star\eqdef\min_{x\in\cX}F(x),\quad
\Delta(x)\eqdef F(x)-F^\star,\\
D_b(x)\eqdef F(x)-\min_{y\in\cX_b}F(x_{-b},y).
\end{gathered}
\end{equation}
Here $\Delta(x)$ is the global objective gap, while $D_b(x)$ is the largest improvement obtainable by modifying block $b$ alone.

\paragraph{Dependency structure.}
The loss need not separate across blocks, since edits to different blocks may affect the same part of the objective.
We capture these interactions through a factorization with known block dependencies.

\begin{assumption}[Dependency Structure]\label{ass:structure}
The objective admits a decomposition
\begin{equation}\label{eq:factors}
F(x)=\sum_{a=1}^{M}\ell_a(x_{N(a)}),
\end{equation}
where $M$ is the number of objective factors and each factor $\ell_a$ depends only on the blocks indexed by the known set $N(a)\subseteq[m]$.
\end{assumption}

This factorization induces a dependency graph $G$ whose vertices correspond to the blocks.
Two blocks $b$ and $c$ are connected by an edge whenever some factor depends on both.
An \emph{independent set} in $G$ therefore contains blocks that never appear together in the same objective factor.
Edits to blocks in an independent set are compatible in the sense that no factor is affected by more than one of them.

We refer to the known factor scopes, or any graph that conservatively encodes them, as a \emph{dependency certificate}.
Such a certificate ensures that blocks in the same independent set do not affect a common objective factor.
In program repair, it may be obtained from declared test--module scopes, interface dependencies, or conservative static analysis.
What matters is shared dependence in the objective rather than physical overlap between the edited components.
Adding extra edges to $G$ preserves the validity of the certificate but may weaken the resulting guarantee.
Omitting a true dependency, however, may incorrectly classify interacting edits as compatible.

For the pricing and checkout example from the introduction, let $x_1,x_2\in\{0,1\}$ indicate whether the pricing and checkout modules apply the discount, and let $x_3\in\{0,1\}$ indicate whether an independent receipt defect is repaired.
Consider
\begin{equation}\label{eq:running}
F(x)=2(1-x_1-x_2)^2+(1-x_3).
\end{equation}
The first term depends jointly on blocks $1$ and $2$, while the second depends only on block $3$.
The resulting dependency graph therefore contains only the edge $1$--$2$.
One valid coloring has color classes $\{1,3\}$ and $\{2\}$, as shown in Figure~\ref{fig:collaboration}.

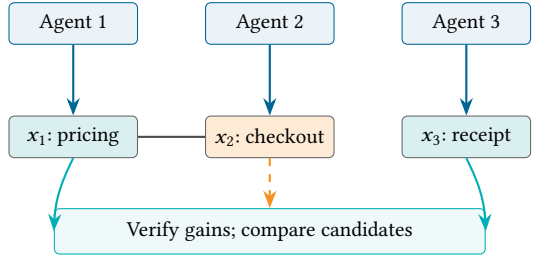
\begin{figure}[t]
\centering
\begin{tikzpicture}[>=Stealth,every node/.style={font=\small},
agent/.style={rounded corners=2pt,draw=MidnightBlue,fill=MidnightBlue!5,minimum width=1.7cm,minimum height=.55cm},
block/.style={rounded corners=2pt,minimum width=1.7cm,minimum height=.55cm,draw=black!60}]
\node[agent] (a1) at (-2.6,1.5) {Agent 1};
\node[agent] (a2) at (0,1.5) {Agent 2};
\node[agent] (a3) at (2.6,1.5) {Agent 3};
\node[block,fill=TealBlue!15] (b1) at (-2.6,0) {$x_1$: pricing};
\node[block,fill=BurntOrange!15] (b2) at (0,0) {$x_2$: checkout};
\node[block,fill=TealBlue!15] (b3) at (2.6,0) {$x_3$: receipt};
\draw[->,MidnightBlue,thick] (a1)--(b1);
\draw[->,MidnightBlue,thick] (a2)--(b2);
\draw[->,MidnightBlue,thick] (a3)--(b3);
\draw[black!70,thick] (b1)--(b2);

\node[rounded corners=2pt,draw=TealBlue,fill=TealBlue!6,minimum width=5.7cm,minimum height=.6cm] (merge) at (0,-1.25) {Verify gains; compare candidates};
\draw[->,TealBlue,thick] (b1.south) to[bend right=12] (merge.west);
\draw[->,TealBlue,thick] (b3.south) to[bend left=12] (merge.east);
\draw[->,BurntOrange,dashed,thick] (b2.south)--(merge.north);
\end{tikzpicture}
\caption{Dependency graph induced by the objective in \eqref{eq:running}.
The discount term depends jointly on pricing and checkout, while the receipt term depends only on the receipt block.
Thus pricing and the receipt can share a color, whereas checkout must receive a different color.}
\Description{Three agents propose edits to pricing, checkout, and receipt blocks. Pricing and checkout are adjacent in the dependency graph, while the receipt is independent of both.}
\label{fig:collaboration}
\end{figure}

\paragraph{Assignments and verified proposals.}
Suppose $n\geq1$ agents are available, indexed by $[n]\eqdef\{1,\ldots,n\}$.
Given a current artifact $x$, an assignment $\cM\subseteq[n]\times[m]$ matches each agent to at most one block and each block to at most one agent.
Each assigned agent receives $x$ together with its designated block, works independently for a time budget $h>0$, and returns a proposed replacement for that block.

Before combining any proposals, the evaluator tests each one separately against the same base artifact $x$, leaving all other blocks unchanged.
A proposal is retained only if it is valid, returned within the time budget, and reduces $F$.
Otherwise, the original block value is retained.
For each assigned pair $(i,b)\in\cM$, let $Y_{ib}$ denote the resulting value of block $b$ after this filtering step.
Its verified singleton gain is
\begin{equation}\label{eq:gain}
g_{ib}\eqdef F(x)-F(x_{-b},Y_{ib})\geq0.
\end{equation}
We call this separate evaluation of each proposal against the common base artifact \emph{singleton verification}.

\section{Verified Proposal Integration}\label{sec:merge}

Singleton verification ensures that every retained proposal improves the common base when applied alone, but it does not determine how those proposals should be combined.

In the running example~\eqref{eq:running}, start from $(0,0,0)$ with loss $3$ and let each agent propose setting its block to one.
The three verified singleton gains are $2$, $2$, and $1$.
Applying all three proposals gives $(1,1,1)$ with loss $2$, while applying only blocks $1$ and $3$ gives loss $0$.
Thus, even accepting the full union whenever it improves $F$ can miss a substantially better compatible subset.

\paragraph{Certified integration.}
Let $\cI(G)$ denote the collection of independent sets of $G$.
A \emph{fractional coloring} $\lambda$ assigns a nonnegative weight $\lambda_I$ to each independent set $I\in\cI(G)$ so that every block is covered by independent sets with total weight at least one.
This fractional-cover view connects to classical work on fractional coloring and covers \citep{lovasz1975ratio,alon1997choosability}; distributed algorithms for fractional coloring are studied by \citet{bousquet2021distributed}.
The total weight of $\lambda$ is $\Lambda(\lambda)$, and the fractional chromatic number $\chi_f(G)$ is the minimum such weight:
\begin{equation}\label{eq:fractional-coloring}
\begin{aligned}
\Lambda(\lambda)
&\eqdef \sum_{I\in\cI(G)}\lambda_I,\\
\chi_f(G)
&\eqdef \min_{\lambda}\Lambda(\lambda)\\
\text{s.t.}\qquad
\lambda_I&\geq0,\qquad I\in\cI(G),\\
\sum_{\substack{I\in\cI(G)\\ b\in I}}\lambda_I
&\geq1,\qquad b\in[m].
\end{aligned}
\end{equation}

To state the integration rule over repeated rounds, we index collaboration rounds by $k$.
At the beginning of round $k$, let $x^k$ denote the common base artifact, let $\cM_k$ denote the assignment, and let $g_{ib}^k$ denote the verified singleton gain of proposal $(i,b)\in\cM_k$ as defined above.

Fix any feasible fractional coloring $\lambda$ and write $\Lambda=\Lambda(\lambda)$.
Only independent sets with $\lambda_I>0$ need to be considered in the merge rule.
For each such set $I$, let $x^{k,I}$ denote the candidate obtained from $x^k$ by applying the retained proposals whose blocks belong to $I$.
Because $I$ is independent, these proposals are compatible and their verified singleton gains add exactly.
We therefore score $I$ by
\begin{equation}\label{eq:fractional-merge-rule}
W_I^k\eqdef
\sum_{\substack{(i,b)\in\cM_k\\ b\in I}}g_{ib}^k.
\end{equation}
The coordinator then selects a highest-scoring independent set among those with $\lambda_I>0$ and applies its proposals:
\begin{equation}
I_k\in\argmax_{I:\lambda_I>0}W_I^k,
\qquad
x^{k,\mathrm{safe}}\eqdef x^{k,I_k}.
\end{equation}
Ties in the choice of $I_k$ are broken by a fixed rule.
The following theorem shows that this rule guarantees a fixed fraction of the total verified singleton gain for every realized collection of agent proposals, regardless of how those proposals were generated.

\begin{theorem}[Certified gain retention]\label{thm:fractional-merge}
Let $\lambda$ be any feasible fractional coloring with total weight $\Lambda=\Lambda(\lambda)$.
Under Assumption~\ref{ass:structure}, for every realization of the agent proposals,
\begin{equation}\label{eq:fractional-merge-bound}
F(x^k)-F(x^{k,\mathrm{safe}})
\geq
\frac{1}{\Lambda}
\sum_{(i,b)\in\cM_k}g_{ib}^k
\geq0.
\end{equation}
In particular, an optimal fractional coloring gives $\Lambda=\chi_f(G)$ and therefore retains at least a $1/\chi_f(G)$ fraction of the total verified singleton gain.
\end{theorem}

\begin{proof}
For readability, suppress the round index $k$ and write $x=x^k$, $\cM=\cM_k$, $g_{ib}=g_{ib}^k$, and $W_I=W_I^k$.
Fix any $I\in\cI(G)$ with $\lambda_I>0$.
For each objective factor $\ell_a(x_{N(a)})$ in~\eqref{eq:factors}, where $a\in[M]$, the independence of $I$ implies $|N(a)\cap I|\leq1$.
Hence each factor is affected by at most one proposal whose block lies in $I$, so the total objective decrease is exactly the sum of the corresponding singleton gains:
\begin{equation}\label{eq:fractional-add}
F(x)-F(x^I)
=
\sum_{\substack{(i,b)\in\cM\\ b\in I}}g_{ib}
=
W_I.
\end{equation}

Since $\sum_{I\in\cI(G)}\lambda_I=\Lambda$ and $\lambda_I\geq0$,
\begin{equation*}
\begin{aligned}
\Lambda\max_{\substack{I\in\cI(G)\\ \lambda_I>0}}W_I
\geq
\sum_{I\in\cI(G)}\lambda_IW_I
=
\sum_{(i,b)\in\cM}g_{ib}
\sum_{\substack{I\in\cI(G)\\ b\in I}}\lambda_I
\geq
\sum_{(i,b)\in\cM}g_{ib}.
\end{aligned}
\end{equation*}
The first inequality follows because the maximum score dominates every $W_I$ with positive weight.
The equality exchanges the order of summation.
The final inequality uses the covering constraint
$\sum_{\substack{I\in\cI(G)\\ b\in I}}\lambda_I\geq1$
and the nonnegativity $g_{ib}\geq0$.

By construction, $I_k$ attains the maximum score among the independent sets with $\lambda_I>0$.
Using~\eqref{eq:fractional-add} for $I_k$ and dividing the preceding inequality by $\Lambda$ gives~\eqref{eq:fractional-merge-bound}.
\end{proof}

Returning to the running example in~\eqref{eq:running}, assigning unit weight to the two color classes gives $\Lambda=2$.
Their scores are $2+1=3$ and $2$, so the certified candidate selects blocks $1$ and $3$ and reaches the optimum.
The guarantee above is also tight in the worst case, as the next theorem shows.

\begin{theorem}[Tightness of certified gain retention]\label{thm:tight}
For any finite nonempty dependency graph $G$, no coordinator that forms its update by selecting a subset of the submitted proposals can guarantee retaining more than a $1/\chi_f(G)$ fraction of the total verified singleton gain in the worst case.
More precisely, for every such $G$, there exists an instance satisfying Assumption~\ref{ass:structure} in which the submitted proposals have total verified gain $\chi_f(G)$, while every subset of them achieves gain at most one.
The bound remains tight when all singleton gains are required to be strictly positive, in which case the worst-case ratio can be made arbitrarily close to $1/\chi_f(G)$.
\end{theorem}

\begin{proof}
The dual of~\eqref{eq:fractional-coloring} maximizes $\sum_{b=1}^m w_b$ over $w_b\geq0$ subject to $\sum_{b\in I}w_b\leq1$ for every $I\in\cI(G)$.
Let $w$ be an optimal dual solution.
By strong duality, $\sum_{b=1}^m w_b=\chi_f(G)$.

Consider binary blocks $z\in\{0,1\}^m$ with initial state $z=0$, and let each block $b$ have a proposal that changes $z_b$ from zero to one.
For any $B>\sum_{b=1}^m w_b$, define
\begin{equation}\label{eq:tight-objective}
F(z)\eqdef C-\sum_{b=1}^m w_bz_b
+B\sum_{\{b,c\}\in E(G)}z_bz_c.
\end{equation}
The interaction terms correspond exactly to the edges of $G$, so this objective has dependency graph $G$.
The singleton gain of block $b$ is $w_b$, and hence the total verified singleton gain is $\chi_f(G)$.

Now consider any subset $S\subseteq[m]$ of the submitted proposals.
If $S$ is independent, its gain is $\sum_{b\in S}w_b\leq1$ by dual feasibility.
If $S$ is not independent, it contains at least one edge of $G$, so at least one penalty term is activated and its gain is at most $\sum_{b=1}^m w_b-B<0$.
Thus every subset achieves gain at most one, while the total singleton gain is $\chi_f(G)$.
Therefore no subset-based coordinator can guarantee a fraction greater than $1/\chi_f(G)$, matching Theorem~\ref{thm:fractional-merge}.

If strictly positive singleton gains are required, replace each $w_b$ by $w_b+\varepsilon$.
The best independent subset then gains at most $1+\varepsilon\alpha(G)$, while the total singleton gain becomes $\chi_f(G)+\varepsilon m$.
Choosing $B$ large enough to keep every non-independent subset unprofitable, the resulting ratio approaches $1/\chi_f(G)$ as $\varepsilon\downarrow0$.
\end{proof}

\paragraph{Ordinary coloring as a practical alternative.}

Ordinary graph coloring gives a directly usable special case.

\begin{corollary}[Ordinary-color guarantee]\label{cor:kappa}
Let $C_1,\ldots,C_\kappa$ be the color classes of any proper $\kappa$-coloring.
Setting $\lambda_{C_c}=1$ for $c\in[\kappa]$ and $\lambda_I=0$ for all other independent sets gives a feasible fractional coloring with $\Lambda=\kappa$.
Rule~\eqref{eq:fractional-merge-rule} then selects the color class with the largest total verified gain.
Writing the resulting candidate as $x^{k,\mathrm{col}}$, we obtain
\begin{equation}\label{eq:kappa-corollary}
F(x^k)-F(x^{k,\mathrm{col}})
\geq
\frac1\kappa\sum_{(i,b)\in\cM_k}g_{ib}^k.
\end{equation}
\end{corollary}

Corollary~\ref{cor:kappa} is obtained by restricting the fractional rule to the $\kappa$ color classes and assigning each unit weight.
For a separable loss, $G$ has no edges, so all retained proposals can be merged simultaneously and $\chi_f(G)=1$.
The fractional-coloring guarantee then 
retains the full sum of the verified gains.

A maximum-weight independent set may improve a round's candidate but must be resolved as gains change.
Fractional coloring is reusable for a fixed graph, though costly to compute; a greedy coloring instead guarantees $1/\kappa$.

\paragraph{Interacting edits.}
The dependency graph certifies when gains are additive, but an edge only indicates that two blocks may interact.
It does not determine whether that interaction is beneficial or harmful.
Two proposals may each repair a local defect and yield an additional improvement when applied together.
To quantify such effects, let $S\subseteq\cM_k$ be a subset of retained proposals and define its \emph{interaction surplus} by
\begin{equation}\label{eq:surplus}
I_k(S)\eqdef F(x^k)-F(x^{k,S})-\sum_{(i,b)\in S}g_{ib}^k,
\end{equation}
where $x^{k,S}$ applies exactly the proposals in $S$.
A positive surplus means that the joint gain exceeds the sum of the singleton gains, while a negative surplus means that it falls short.
If the blocks represented in $S$ form an independent set in $G$, then $I_k(S)=0$.
In the running example \eqref{eq:running}, the full union has gain $1$ while the singleton gains sum to $5$, giving interaction surplus $-4$.

Independent-set integration avoids such harmful interactions, but it can also miss beneficial ones.
Let $x^{k,\mathrm{all}}$ denote the candidate obtained by applying all retained proposals to the common base.
To capture favorable interactions that arise in this full union without sacrificing the certified guarantee, the \emph{optimistic} rule compares it with the safe candidate:
\begin{equation}\label{eq:optimistic}
x^{k+1}\in\operatorname*{arg\,min}_{z\in\{x^{k,\mathrm{safe}},x^{k,\mathrm{all}}\}}F(z).
\end{equation}
The safe candidate is chosen in the event of a tie and whenever the union is invalid.
Its objective value is already determined by additivity, so only the union requires an additional evaluation.

\begin{proposition}[Optimistic integration]\label{prop:optimistic}
Under Assumption~\ref{ass:structure} and exact, side-effect-free candidate evaluation, rule~\eqref{eq:optimistic} satisfies
\begin{equation}\label{eq:optimistic-bound}
F(x^k)-F(x^{k+1})
=
\max\{G_{\mathrm{safe}}^k,G_{\mathrm{all}}^k\}
\geq
G_{\mathrm{safe}}^k
\geq
\frac1\Lambda\sum_{(i,b)\in\cM_k}g_{ib}^k,
\end{equation}
where $G_{\mathrm{safe}}^k\eqdef F(x^k)-F(x^{k,\mathrm{safe}})$ and $G_{\mathrm{all}}^k\eqdef F(x^k)-F(x^{k,\mathrm{all}})$, with $G_{\mathrm{all}}^k=-\infty$ when the union is invalid.
\end{proposition}

\begin{proof}
Rule~\eqref{eq:optimistic} selects the candidate with the larger gain, so its gain is at least $G_{\mathrm{safe}}^k$.
Theorem~\ref{thm:fractional-merge} gives the final inequality.
\end{proof}

Proposition~\ref{prop:optimistic} therefore shows that the optimistic rule preserves the same per-round guarantee as the safe rule while allowing the full union to be selected when it performs better.
This comparison is local to the current round, since the additional evaluation and the choice of the accepted artifact can affect the subsequent trajectory.

\paragraph{One round of the method.}
Algorithm~\ref{alg:round} summarizes one round of the optimistic integration method.

\begin{algorithm}[t]
\caption{Optimistic Integration}
\label{alg:round}
\begin{algorithmic}[1]
\Require Base artifact $x^k$, fractional coloring $\lambda$, deadline $h$
\Ensure Updated artifact $x^{k+1}$
\State Choose an assignment $\cM_k$ and send each assigned agent $x^k$ and its designated block
\State Collect proposals returned by deadline $h$
\State Verify each proposal separately against $x^k$ and compute the gains $g_{ib}^k$
\ForAll{$I$ with $\lambda_I>0$}
    \State $W_I^k \gets \sum_{(i,b)\in\cM_k:\,b\in I} g_{ib}^k$
\EndFor
\State $I_k \gets \argmax_{I:\lambda_I>0} W_I^k$
\State $x^{k,\mathrm{safe}} \gets x^{k,I_k}$
\State Construct the full-union candidate $x^{k,\mathrm{all}}$
\If{$x^{k,\mathrm{all}}$ is valid}
    \State $x^{k+1} \gets
    \arg\min_{z\in\{x^{k,\mathrm{safe}},x^{k,\mathrm{all}}\}} F(z)$
\Else
    \State $x^{k+1} \gets x^{k,\mathrm{safe}}$
\EndIf
\State Broadcast $x^{k+1}$
\end{algorithmic}
\end{algorithm}

The safe-only variant omits the full-union evaluation and sets $x^{k+1}=x^{k,\mathrm{safe}}$.
In the ordinary-color specialization of Corollary~\ref{cor:kappa}, the supported independent sets are simply the color classes of the chosen proper coloring.
The deadline $h$ applies to local proposal generation, while verification and communication contribute additional round time.

\section{Global Convergence from Local Progress}\label{sec:theory}

The per-round integration guarantee depends only on the verified singleton gains and the dependency structure.
To obtain a global convergence rate, we introduce three additional ingredients: local solver quality, a block error bound, and assignment coverage.
We then combine these factors with the integration guarantee to obtain the contraction rate.

\paragraph{Local solver quality.}
Let $\cH_k$ denote all information available before round $k$, including the current artifact, previous observations, and any context retained by the agents.
We model an assigned agent as an inexact block solver and measure its quality relative to $D_b(x^k)$, the best improvement available by changing block $b$ alone.

\begin{assumption}[Conditional local improvement]\label{ass:oracle}
For known lower bounds $\eta_{ib}(h)\in[0,1]$, the following holds at every round, every history reached by the protocol, and every assignment with positive conditional probability:
\begin{equation}\label{eq:oracle}
\ExpCond{g_{ib}^k}{\cH_k,\cM_k}
\geq
\eta_{ib}(h)D_b(x^k),
\qquad
(i,b)\in\cM_k.
\end{equation}
\end{assumption}

The parameter $\eta_{ib}(h)$ quantifies the expected quality of agent $i$ on block $b$ within deadline $h$ relative to the best available single-block improvement.
For example, $\eta_{ib}(h)=1/2$ means that the agent recovers, in expectation, at least half of $D_b(x^k)$, rather than solving half of its assigned tasks exactly.
Because the expectation includes timeouts and imperfect edits, these outcomes are already reflected in the guarantee.
The bound is conditioned on both the history and the full assignment, so agent outputs may be arbitrarily correlated.
At the extremes, $\eta_{ib}(h)=0$ imposes no capability requirement on the agent--block pair, whereas any positive uniform lower bound represents a substantive assumption on its performance.

\paragraph{Block error bound.}
Local solver quality alone is not sufficient for global progress, since even an exact block solver cannot improve a state when no individual block admits an improving move.
We therefore require the available single-block improvements to capture a nonvanishing fraction of the remaining global error.

\begin{assumption}[Block error bound]\label{ass:error}
There exists $\mu\in(0,1]$ such that, for every $k\geq0$, the iterates generated by the protocol satisfy almost surely
\begin{equation}\label{eq:error}
\sum_{b=1}^m D_b(x^k)\geq \mu\Delta(x^k).
\end{equation}
\end{assumption}

The parameter $\mu$ measures how much of the remaining objective gap is visible through improvements to individual blocks.
For a separable loss, the blockwise improvements add exactly to the global gap, so $\mu=1$.
More generally, a positive $\mu$ rules out suboptimal states at which every individual block is already locally optimal.

The role of this condition is illustrated by a simple two-block example.
Let $x_1,x_2\in\{0,1\}$ and consider
\begin{equation}\label{eq:trap}
F(x_1,x_2)\eqdef1+x_1+x_2-3x_1x_2.
\end{equation}
At $(0,0)$, the loss is $1$.
Changing either block alone increases the loss to $2$, whereas changing both blocks yields the global minimum $0$.
Hence $D_1(0,0)=D_2(0,0)=0$ while $\Delta(0,0)=1$, so no positive $\mu$ can satisfy~\eqref{eq:error} at this state.
This obstruction arises from the task structure rather than the quality of the local solver.
A method restricted to individually improving block updates cannot leave $(0,0)$, and progress requires either a joint update or acceptance of an intermediate increase in the objective.

\paragraph{Assignment coverage.}
Local solver quality translates into global progress only if useful agents are assigned to each block sufficiently often.
Let $\mathfrak M$ denote the finite set of agent--block matchings.
Fix a distribution $p=(p_{\cM})_{\cM\in\mathfrak M}$ and assume that, at every round, the matching is drawn according to
\begin{equation}
\Pr(\cM_k=\cM\mid\cH_k)=p_{\cM},
\qquad
\sum_{\cM\in\mathfrak M}p_{\cM}=1.
\end{equation}
For each agent--block pair, define its marginal assignment probability by
\begin{equation}\label{eq:coverage}
q_{ib}\eqdef\sum_{\cM\ni(i,b)}p_{\cM},
\qquad
a_b\eqdef\sum_iq_{ib}\eta_{ib}(h),
\qquad
a\eqdef\min_ba_b.
\end{equation}
The quantity $a_b$ combines how often block $b$ is assigned to agent $i$ with the corresponding local-solver quality $\eta_{ib}(h)$.
Equivalently, it measures the guaranteed fraction of the best block-$b$ improvement recovered on average under the assignment distribution.
The quantity $a$ is the corresponding worst-case coverage across blocks.
Since a matching assigns at most one agent to each block and $\eta_{ib}(h)\leq1$, we have $0\leq a_b\leq1$ and hence $0\leq a\leq1$.
If $a=0$, then at least one block receives no guaranteed improvement under the assignment scheme, so these bounds cannot certify positive worst-case progress.

\paragraph{Combining the factors.}
The preceding quantities combine directly with the per-round integration guarantee.
Under the assignment distribution and Assumption~\ref{ass:oracle}, the expected verified gain is at least an $a$ fraction of the total available blockwise improvement.
Assumption~\ref{ass:error} lower bounds this quantity by a $\mu$ fraction of the remaining objective gap, while the integration rule retains at least a $1/\Lambda(\lambda)$ fraction of the verified gain.
We therefore define
\begin{equation}\label{eq:theta}
\theta_\lambda\eqdef\frac{\mu a}{\Lambda(\lambda)}.
\end{equation}
Thus, in expectation, each round reduces the current objective gap by at least a $\theta_\lambda$ fraction.
For an optimal fractional coloring, $\Lambda(\lambda)=\chi_f(G)$ and hence $\theta_\lambda=\mu a/\chi_f(G)$.
Using a proper $\kappa$-coloring instead gives the corresponding factor $\mu a/\kappa$.

\begin{theorem}[Geometric convergence]\label{thm:rate}
Fix either the safe-only or optimistic integration rule, a deadline $h$, a feasible fractional coloring $\lambda$, and the assignment distribution $p$ above.
Suppose Assumptions~\ref{ass:structure}--\ref{ass:error} hold along the resulting trajectory.
Then, for deterministic $x^0$ and every integer $K\geq0$,
\begin{equation}\label{eq:rate}
\Exp{\Delta(x^K)}
\leq
(1-\theta_\lambda)^K\Delta(x^0).
\end{equation}
\end{theorem}

The conditional-expectation proof, given in the supplementary material, formalizes the preceding argument and yields geometric decay of the expected objective gap with contraction factor $1-\theta_\lambda$.
The guarantee is an average one across the randomness of the protocol and does not require every realized round to achieve the same fractional reduction.
Through $\theta_\lambda=\mu a/\Lambda(\lambda)$, the rate separates the effects of task structure, agent quality and assignment coverage, and proposal integration.
This decomposition also provides the basis for the next section, where we choose assignments by balancing their expected progress against their time and token costs.

\section{Work and Cost Allocation}\label{sec:allocation}

For a fixed task, fractional coloring, and deadline, the assignment policy can improve the coverage factor $a$ by directing useful agents toward poorly covered blocks.
Optimizing this policy is therefore a separate problem over assignments.

\paragraph{Assignment optimization.}
Fix a deadline $h$ and a per-round expected token budget $b_0$.
Let $C(\cM,h)$ denote an upper bound on the token cost of using matching $\cM$ at deadline $h$, including communication and verification.
Given the capability bounds $\eta_{ib}(h)$, we choose the assignment distribution to maximize the weakest block coverage:
\begin{equation}\label{eq:lp}
\begin{aligned}
\max_{p,a}\quad &a\\
\text{subject to}\quad
&\sum_{\cM}p_{\cM}\sum_{i:(i,b)\in\cM}\eta_{ib}(h)\geq a
&&\forall b,\\
&\sum_{\cM}p_{\cM}C(\cM,h)\leq b_0,\\
&\sum_{\cM}p_{\cM}=1,
\qquad
p_{\cM}\geq0.
\end{aligned}
\end{equation}
The coverage constraints ensure $a_b\geq a$ for every block, while the budget constraint limits the expected token cost.
For fixed $h$ and capability bounds $\eta_{ib}(h)$, the linear program maximizes the coverage factor in the certified convergence rate over fixed assignment distributions.

\paragraph{Rate--cost trade-off.}
Improving the per-round convergence rate does not necessarily reduce the total resources required to reach a target error.
For example, increasing the deadline $h$ can improve the capability bounds $\eta_{ib}(h)$ and hence the coverage factor $a$, while also increasing the duration and token cost of each round.
A similar trade-off arises with optimistic integration, where evaluating the full union adds per-round cost but may reduce the number of rounds when the union yields greater progress.

Let $d_{\lambda}(h)$ bound the duration of a safe-only round, including proposal generation, verification, communication, and other overhead, and let $\tau_{\mathrm{union}}(h)$ bound the additional time required to evaluate the full union.
For optimistic integration, define
\begin{equation}\label{eq:optimistic-cost}
\begin{gathered}
d(h)\eqdef d_{\lambda}(h)+\tau_{\mathrm{union}}(h),\\
C(\cM,h)\eqdef C_{\lambda}(\cM,h)+C_{\mathrm{union}}(\cM,h),
\end{gathered}
\end{equation}
where $C_{\lambda}(\cM,h)$ and $C_{\mathrm{union}}(\cM,h)$ are the corresponding token-cost bounds.
For an assignment distribution $p$, the expected token cost per round is
\begin{equation}
\overline C(h)\eqdef\sum_{\cM}p_{\cM}C(\cM,h).
\end{equation}

Suppose $a>0$ and let $0<\varepsilon<\Delta(x^0)$ be a target expected error.
By Theorem~\ref{thm:rate} and $1-\theta_\lambda\leq e^{-\theta_\lambda}$,
\[
\Exp{\Delta(x^K)}
\leq
e^{-K\theta_\lambda}\Delta(x^0).
\]
A sufficient number of rounds to ensure $\Exp{\Delta(x^K)}\leq\varepsilon$ is therefore
\begin{equation}\label{eq:horizon}
K_\varepsilon
\eqdef
\left\lceil
\frac{\Lambda(\lambda)}{\mu a}
\log\frac{\Delta(x^0)}{\varepsilon}
\right\rceil.
\end{equation}
Let $T_\varepsilon$ denote the corresponding elapsed-time bound and $\overline C_\varepsilon$ the expected total token-cost bound.
Then
\begin{equation}\label{eq:resource-bounds}
T_\varepsilon
\leq
d(h)K_\varepsilon,
\qquad
\overline C_\varepsilon
\leq
\overline C(h)K_\varepsilon.
\end{equation}

These bounds make the deadline trade-off explicit.
Increasing $h$ may reduce $K_\varepsilon$ by improving the capability bounds and hence the coverage factor $a$, while simultaneously increasing the per-round duration $d(h)$ and expected token cost $\overline C(h)$.
Candidate deadlines can therefore be compared using the total-time bound $d(h)K_\varepsilon$ together with the corresponding token-cost constraint.

\section{Experiments}\label{sec:experiments}

We examine whether verified integration resolves conflicts among local repairs, when parallel work offsets its coordination cost, and whether fractional coloring retains more gain from the same proposals.
The four repair benchmarks use Python tasks with declared module dependencies and an objective $F$ that counts failed checks.
Agents see the specification, current code, and visible checks.
The coordinator uses visible and hidden checks to determine acceptance and stopping. A run is successful when its final artifact passes all checks.
We also examine proposal integration on public road networks and course-timetabling instances, and the scheduling of parallel route searches on vehicle-routing instances.

\paragraph{Protocols and evaluation.}
For the conflict and scale benchmarks, we compare safe color-class integration (\textsf{COL}) and optimistic integration (\textsf{OPT}) with four baselines.
\textsf{MA} applies all parallel proposals, while \textsf{MAG} accepts their union only when it lowers $F$.
\textsf{SGL} assigns the whole program to one agent, and \textsf{SEQ} visits modules sequentially. Both baselines accept only improving edits.
All protocols start from the same faulty artifact.
Each call starts with a fresh context, and only edits to its assigned modules are kept.
The parallel rules share first-round proposals and subsequently work from their own accepted artifacts.
Appendix~\ref{app:bench} provides further experimental details.

We report input-token-equivalent cost, which weights ordinary input, cache-write, and cache-read tokens and excludes output tokens.
For modeled completion time, each round contributes the duration of its slowest agent call, together with an assumed verification latency $t_{\mathrm{ver}}$ for each verification stage that must complete before the procedure can continue.
Checks performed in parallel within the same stage incur this latency only once.
Agent calls within a round are assumed to start simultaneously.
The accounting formula and sensitivity to $t_{\mathrm{ver}}$ are given in Appendices~\ref{app:bench}--\ref{app:scale}.

\subsection{Repairing conflicting edits}

The first benchmark contains seven three-module tasks whose dependency graphs are paths, with $\chi_f(G)=\chi(G)=2$.
In six tasks, either of two adjacent modules can perform a required transformation, but applying it twice causes an error.
The seventh task states interface ownership explicitly.
We run three replicates per task with \texttt{claude-haiku-4-5} and one replicate with \texttt{claude-sonnet-5}, giving 21 and seven task--replicate units, respectively.
To diagnose a conflict, we test the two adjacent proposals while keeping the remaining module fixed to a known-good implementation.
We classify the two proposals as conflicting when each passes the interface check individually, but applying both together causes the check to fail.

\begin{table}[t]
\caption{Success rates on the repair benchmark.
The upper panel groups runs by model.
The lower groups the same runs by whether the two adjacent proposals conflict in the first round.
Each entry reports successful runs over total runs.}
\label{tab:repair}
\centering
\small
\renewcommand{\arraystretch}{1.10}
\begin{tabular*}{\columnwidth}{@{\extracolsep{\fill}}lcccccc@{}}
\toprule
 & \textsf{SGL} & \textsf{SEQ} & \textsf{MA} & \textsf{MAG} & \textsf{COL} & \textsf{OPT}\\
\midrule
\multicolumn{7}{@{}l}{\textit{By model}}\\
Haiku & 21/21 & 19/21 & 12/21 & 13/21 & 21/21 & 21/21\\
Sonnet & 7/7 & 7/7 & 5/7 & 5/7 & 7/7 & 7/7\\
\addlinespace[2pt]
\multicolumn{7}{@{}l}{\textit{By first-round conflict}}\\
Conflict & 11/11 & 11/11 & 0/11 & 1/11 & 11/11 & 11/11\\
No conflict & 17/17 & 15/17 & 17/17 & 17/17 & 17/17 & 17/17\\
\bottomrule
\end{tabular*}
\end{table}

Table~\ref{tab:repair} shows that both verified rules, \textsf{COL} and \textsf{OPT}, solve all 28 task--replicate units.
\textsf{MA} solves all 17 units without a first-round conflict but none of the 11 conflict units, while \textsf{MAG} recovers only one of these conflicts.
Task T1 illustrates this failure mode.
Under \textsf{MA}, the two agents repeatedly add or remove the same transformation, producing a failure-count trajectory of $21,5,5,5$.
\textsf{OPT} verifies the removals separately and accepts one, leaving a single copy of the transformation and reducing the trajectory to $21,5,0$.
Across both models, \textsf{OPT} selected the color candidate in 11 of 37 union evaluations, and all 108 audited color merges satisfied the additivity condition.

These small tasks favor the single-agent baseline in resource use.
\textsf{SGL} solves all 28 units with a median of one agent call.
For \texttt{claude-haiku-4-5}, its input cost is roughly one quarter that of \textsf{COL} and \textsf{OPT}, while its modeled completion time is about one third as large when $t_{\mathrm{ver}}=0$.
This additional coordination overhead becomes consequential under tight budgets.
Under an input-cost budget of nine units, where one unit is the median input-token-equivalent cost of a \texttt{claude-haiku-4-5} agent call and each objective evaluation is charged one additional unit, \textsf{MA} solves 12 of 21 runs, compared with eight for both \textsf{COL} and \textsf{OPT}.
Appendix~\ref{app:results} reports the complete budget comparisons and representative failure traces.

\subsection{Scaling compatible parallel work}

The second benchmark increases the amount of available parallel work using twelve \texttt{claude-haiku-4-5} instances with 4--32 modules, with three instances at each size.
Each module contains five bugs and an explicit interface contract.
The resulting dependency graphs are sparse and bipartite, with $\chi_f(G)=\chi(G)=2$.
To isolate the effect of the coordination protocol, agents are not allowed to delegate work further.
All protocols solve every instance on which they are run, with \textsf{SEQ} omitted for the 32-module instances.
\textsf{SGL} completes every instance in a single session, and no genuine proposal conflicts are observed.

\begin{table}[t]
\caption{Scaling results as the number of modules $m$ increases.
Entries are geometric means over three instances at each size.
Input cost excludes output tokens and verification charges.}
\label{tab:scale}
\centering
\small
\setlength{\tabcolsep}{3pt}
\renewcommand{\arraystretch}{1.10}
\begin{tabular*}{\columnwidth}{@{\extracolsep{\fill}}rcccccc@{}}
\toprule
& \multicolumn{2}{c}{Completion time, $t_{\mathrm{ver}}=0$ (min)}
& \multicolumn{3}{c}{\textsf{SGL}/\textsf{OPT} speedup}
& \textsf{OPT}/\textsf{SGL} \\
\cmidrule(lr){2-3}
\cmidrule(lr){4-6}
$m$
& \textsf{SGL}
& \textsf{OPT}
& $t_{\mathrm{ver}}=0$
& $60$
& $300$
& Input cost \\
\midrule
4  & 2.4  & 3.5 & 0.68 & 0.67 & 0.67 & 2.23 \\
8  & 4.2  & 2.9 & 1.43 & 1.04 & 0.79 & 2.23 \\
16 & 9.2  & 4.5 & 2.05 & 1.43 & 0.91 & 1.94 \\
32 & 15.3 & 4.1 & 3.69 & 2.42 & 1.32 & 2.02 \\
\bottomrule
\end{tabular*}
\end{table}

Table~\ref{tab:scale} shows that the benefit of parallel execution grows with task size.
At 32 modules, \textsf{OPT} reduces the modeled completion time from 15.3 minutes for \textsf{SGL} to 4.1 minutes when $t_{\mathrm{ver}}=0$, while using $2.02\times$ the input cost.
With verification latencies of 60 and 300 seconds, the corresponding speedups decrease from $3.69\times$ to $2.42\times$ and $1.32\times$.
Since no proposal conflicts arise in this benchmark, \textsf{OPT} never invokes its fallback rule.
When $t_{\mathrm{ver}}=0$, \textsf{MA} and \textsf{OPT} have the same modeled completion time.
For $t_{\mathrm{ver}}>0$, \textsf{OPT} is slower because of its additional verification steps.
Appendix~\ref{app:scale} reports the instance-level results.

\subsection{Comparing ordinary and fractional coloring}

\begin{figure}[htbp]
\centering
\includegraphics[width=\columnwidth]{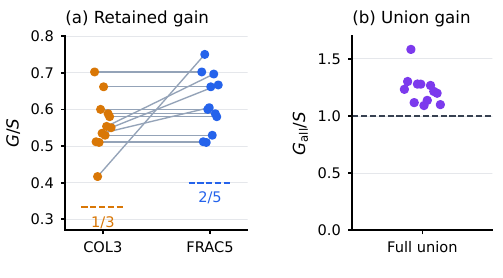}
\caption{Retained gain and union interactions on the twelve $C_5$ units, normalized by the sum $S$ of positive verified singleton gains.
(a) Paired gains for ordinary (\textsf{COL3}) and fractional (\textsf{FRAC5}) coloring; lines join identical proposal sets. Dashed segments mark the guarantees $1/3$ and $2/5$.
(b) Directly evaluated full-union gain; the dashed line at $1$ is the additive baseline.}
\Description{C5-only two-panel diagnostic. The left panel connects ordinary and fractional retained gains for each of twelve paired proposal sets. The right panel shows full-union gains normalized by positive singleton gains.}
\label{fig:graph-diagnostic}
\end{figure}

To isolate the effect of fractional coloring, the third benchmark uses $C_5$, for which $\chi(C_5)=3$ and $\chi_f(C_5)=5/2$.
We construct twelve units spanning four domains and three rotations of an ordinary coloring, fixed before the runs.
Each unit assigns one fresh agent to each of the five modules.
Using the same first-round proposals, \textsf{COL3} chooses among the three classes of the fixed ordinary coloring, while \textsf{FRAC5} evaluates all five independent pairs.
The corresponding guaranteed retained fractions are $1/3$ and $2/5$.

On identical proposal sets, \textsf{FRAC5} never retains less verified gain than \textsf{COL3} and improves over it in four of the twelve units.
The mean retained fraction increases from $0.5615$ to $0.6167$ (Figure~\ref{fig:graph-diagnostic}).
The full union gains more than the sum of the verified singleton gains in all twelve units, so optimistic integration selects the union throughout this benchmark.
Appendix~\ref{app:fractional} gives the complete benchmark and comparison details.

\subsection{Repairing held-out tasks over multiple rounds}

We next examine whether the benefit of verified integration persists over several repair rounds.
The fourth benchmark contains sixteen held-out interface-migration designs with five or ten modules, arranged on $C_5$ or prism graphs.
Three replicates per design give 48 task--replicate units.
With \texttt{Qwen2.5-Coder-32B-Instruct}, we compare the four parallel rules above with safe fractional integration (\textsf{FRAC}) and its optimistic variant (\textsf{FOPT}).
All six rules receive the same first-round proposals and then continue from their own accepted programs for at most three rounds.

\textsf{FOPT} reduces the mean number of failed checks from 34.38 to 6.29, compared with 20.02 for \textsf{MA}.
Its final loss is lower on every task--replicate unit and every design mean (Figure~\ref{fig:heldout-paired}).
The first-round fractional candidate also has lower loss than the full union on 47 of the 48 units, showing how selecting compatible edits avoids failures caused by their combination.
None of the six parallel rules fully solves a task within three rounds.

Fractional coloring retains more gain than ordinary coloring on 30 of the 48 first-round proposal sets.
The difference is smaller at termination: \textsf{FRAC} and \textsf{COL} leave means of 6.17 and 6.19 failed checks.
Later proposals depend on the program accepted by each rule, so the initial advantage need not persist through subsequent rounds.
A single agent restricted to one module per call leaves a mean of 12.27 failed checks; allowing multi-file edits and up to twice \textsf{FOPT}'s token use gives 20.52, with frequent output-parsing failures.
Appendix~\ref{app:interface} gives the control protocols, resource accounting, and results by graph family.

\begin{figure}[t]
\centering
\includegraphics[width=\columnwidth]{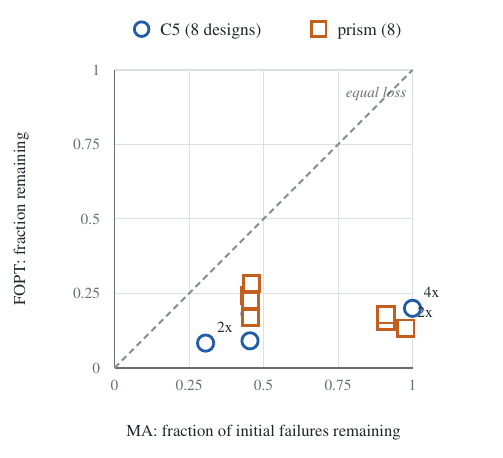}
\caption{Remaining failed-check fractions on sixteen held-out interface-migration designs after up to three rounds. Each point is the mean final loss over three replicates divided by the initial loss. Circles denote $C_5$ and squares denote prism graphs; numerals identify coincident points. Points below the diagonal favor \textsf{FOPT} over \textsf{MA}.}
\Description{Paired scatterplot of normalized final failed-check counts. All sixteen design points lie below the diagonal of equal loss, indicating fewer remaining failures for FOPT than for merge-all.}
\label{fig:heldout-paired}
\end{figure}

\subsection{Integrating proposals on road networks}\label{sec:road-merger}

We test the merger mechanism on factors taken directly from three public road networks~\cite{snap-roadnet}. A deterministic, proposal-blind structural screen selected 30 nonoverlapping held-out regions, ten per network, each containing five connected blocks of 32 road vertices whose \emph{actual interblock roads} form an induced $C_5$. Three fixed initial-label seeds per region yield 90 paired region--seed units. Each road edge touching a selected block contributes one factor to a binary graph-cut objective. Five separate CPU agents each make 30,000 balanced-swap attempts within one block. A common guarded-union warm start precedes the comparison. All integration rules then score the same five fresh, individually verified proposals from the same incumbent. Appendix~\ref{app:road-merger} gives the frozen screen and protocol.

\begin{table}[t]
\caption{Held-out road-network proposal integration: sums of graph-cut edges removed over 30 paired region--seed units per network. Larger is better. \textsf{COL} uses one ordinary coloring fixed before proposals; \textsf{FRAC} considers the five supported $C_5$ independent pairs.}
\label{tab:road-merger}
\centering
\small
\setlength{\tabcolsep}{4pt}
\begin{tabular}{lrrrr}
\toprule
Network & \textsf{COL} & \textsf{FRAC} & \textsf{MAG} & \textsf{MA} \\
\midrule
Pennsylvania & 283 & 343 & 68 & 14 \\
California & 278 & 305 & 114 & 87 \\
Texas & 238 & 274 & 60 & 28 \\
\midrule
Total & 799 & \textbf{922} & 242 & 129 \\
\bottomrule
\end{tabular}
\end{table}

The fractional-support rule retains 922 of 1,617 verified singleton-gain units (57.0\%), compared with 799 (49.4\%) for the pre-fixed ordinary coloring, a 15.4\% relative increase (Table~\ref{tab:road-merger}). It is strictly better in 52 of 90 paired runs and ties in the rest; no loss is possible here by construction because its support dominates the pre-fixed color classes for nonnegative gains. The three-seed region means favor it in 28 of 30 regions and tie in two. Unguarded full union increases the graph-cut objective in 31 runs despite every nonnull proposal improving the common incumbent alone.

Adaptive ordinary coloring, chosen after observing singleton gains, can include the highest-gain independent pair and also retains 922 units across all 90 runs. Fractional and adaptive ordinary coloring therefore attain the same total gain, exceeding that of the coloring fixed in advance.

\subsection{Integrating proposals in course timetabling}

We evaluate verified integration on nine ITC-2007 timetabling instances~\cite{digaspero2007itc}, selected for induced $C_5$ course subgraphs before proposal generation. Five computational agents independently change lecture periods for separate courses from a common timetable, keeping rooms fixed; dependencies follow shared teachers, curricula, or fixed rooms. On identical proposals, \textsf{FRAC} retains 10\% more improvement than fixed \textsf{COL}, winning on two instances and tying on seven (Table~\ref{tab:itc-merger}). All 45 compatible-pair merges satisfy exact additivity, and the selected compatible updates pass the official validator. \textsf{MA} creates a room conflict on \texttt{comp13}; adaptive ordinary coloring matches \textsf{FRAC}.

\begin{table}[t]
\caption{Penalty reductions and feasible outcomes on nine ITC-2007 instances. Higher reduction is better. The \textsf{MA} total is omitted because one timetable is infeasible.}
\label{tab:itc-merger}
\centering
\small
\begin{tabular}{lcc}
\toprule
Method & Penalty reduction & Feasible \\
\midrule
Fixed \textsf{COL} & 150 & $9/9$ \\
\textsf{FRAC} & \textbf{165} & $9/9$ \\
\textsf{MA} & --- & $8/9$ \\
\bottomrule
\end{tabular}
\end{table}

\subsection{Scheduling parallel route searches}

We test a different use of fractional coloring: scheduling preemptible search by agents that hold exclusive locks on their assigned pairs of vehicle routes. On each of four public instances of the capacitated vehicle-routing problem (CVRP), five stochastic local-search agents are assigned consecutive pairs of routes, deliberately inducing a $C_5$ conflict graph. Nonadjacent agents can search concurrently on two CPU workers. Both schedules allocate exactly 500,000 candidate-move evaluations to every agent. A fixed ordinary coloring executes the compatible groups $(0,2)$, $(1,3)$, $(4)$ in three full-budget waves; a fractional cover executes $(0,2)$, $(1,3)$, $(2,4)$, $(3,0)$, $(4,1)$ in five half-budget waves. Thus $\chi(C_5)=3$ and $\chi_f(C_5)=5/2$ predict a shorter critical path under this locking and preemption model.

\begin{table}[t]
\caption{Wall-time reduction under fractional route-search scheduling. Each instance has three paired repeats. ``Faster pairs'' counts runs in which the fractional schedule is faster than the ordinary schedule. The final column gives the median paired reduction.}
\label{tab:routing-schedule}
\centering
\small
\begin{tabular}{lcc}
\toprule
Public CVRP instance & Faster pairs & Median wall reduction \\
\midrule
Loggi-n401-k23 & $3/3$ & $14.74\%$ \\
Loggi-n501-k24 & $3/3$ & $15.27\%$ \\
ORTEC-n242-k12 & $3/3$ & $13.52\%$ \\
ORTEC-n323-k21 & $3/3$ & $15.76\%$ \\
\midrule
All pairs & $12/12$ & $15.00\%$ \\
\bottomrule
\end{tabular}
\end{table}

Table~\ref{tab:routing-schedule} reports measured wall time, rather than the modeled latency used in the repair scaling benchmark. Across the twelve pairs, fractional scheduling is faster by 13.16--17.28\%; the median reduction is 15.00\%. Summed agent CPU time is nearly unchanged (fractional/integral median ratio 1.008). All 24 final route sets satisfy capacity and customer-coverage constraints, but route cost is better under the fractional schedule in only four pairs, tied in five, and worse in three. Appendix~\ref{app:routing-schedule} gives the full paired times and route-cost differences.

The $C_5$ was induced by a deterministic route assignment, not discovered as a natural dependency graph. Two inputs were used during development and the next two were selected after observing initial outcomes. The speedup depends on holding route-pair locks throughout search chunks. An immutable-snapshot system that parallelizes all searches and serializes only commits is outside this comparison.

\FloatBarrier
\section{Conclusion}
We formulate parallel agent collaboration as verified proposal integration. Evaluating proposals against a common artifact yields local gains, while the dependency graph identifies compatible sets whose gains add exactly. Fractional coloring gives the tight worst-case retention factor $1/\chi_f(G)$, and evaluating the full union allows the coordinator to capture favorable interactions. Under the local-improvement and block error-bound assumptions, repeated integration converges geometrically. The rate separates task structure, agent allocation, and proposal integration, providing a basis for choosing assignments and balancing progress against time and cost.

The experiments show how compatibility shapes both the quality and cost of parallel work. Selecting compatible proposals resolves conflicts between individually useful repairs and improves retained gain over a fixed coloring on road-network objectives. Adaptive ordinary coloring matches the latter result, linking the improvement to the compatible sets considered. Parallel execution reduces modeled completion time on larger compatible repair tasks at higher input cost, while fractional scheduling shortens measured route-search time under exclusive route access. These results connect the integration of local edits and the scheduling of parallel work through their shared dependency structure.

\label{main-text-end}

\begin{acks}
The research reported in this publication was supported by funding from King Abdullah University of Science and Technology (KAUST): i) KAUST Baseline Research Scheme, ii) CRG Grant ORFS-CRG12-2024-6460, and iii) Center of Excellence for Generative AI, under award number 5940.
\end{acks}

\clearpage
\onecolumn
\section*{Appendix}
\label{supplement-start}

\appendix
The appendix gives the experimental protocols and detailed results, followed by controlled examples of allocation and merging and the convergence proof.

\section*{AI use statement}
We used OpenAI ChatGPT/Codex and Anthropic Claude to discuss mathematical arguments and experimental design, implement experiments, check literature, and edit the manuscript.
The authors reviewed AI-assisted content, checked the proofs and citations, and validated the reported experimental results.

\section{Repair benchmark and evaluation protocols}\label{app:bench}

\begin{table}[h]
\caption{Tasks in the three-module repair benchmark. Each task has a path dependency graph. ``Lines'' counts faulty editable code; ``Checks'' gives visible/hidden check counts; and $F(x^0)$ is the initial number of failures. ``Double'' counts the integration failures caused by applying the interface transformation in both adjacent modules when all other code is correct.}
\label{tab:tasks}
\centering
\footnotesize
\setlength{\tabcolsep}{3pt}
\begin{tabular}{llllrrrr}
\toprule
Task & Modules (blocks) & Ambiguous interface & Edge & Lines & Checks & $F(x^0)$ & Double\\
\midrule
T1 invoice & parsing, pricing, checkout & coupon discount & pricing--checkout & 44 & 25/10 & 21 & 5\\
T2 profiles & users, cards, pages & HTML escaping & users--cards & 24 & 18/8 & 19 & 4\\
T3 calendar & timeparse, planner, agenda & UTC offset & timeparse--planner & 49 & 23/10 & 26 & 5\\
T4 CSV export & records, csvrow, exporter & CSV quoting & records--csvrow & 20 & 19/9 & 24 & 6\\
T5 blog URLs & slugs, urls, sitemap & percent-encoding & slugs--urls & 23 & 18/9 & 25 & 5\\
T6 logging & events, formatter, sink & log escaping & formatter--sink & 26 & 20/8 & 26 & 6\\
C0 pilot (control) & parse, total, receipt & none & --- & 9 & 24/6 & 28 & ---\\
\bottomrule
\end{tabular}
\end{table}

\paragraph{Task construction.}
The seven tasks contain 26--35 checks and 9--49 lines of faulty editable code each, with 207 checks in total. Each task provides a specification, three faulty modules, and 13--19 lines of fixed code connecting them. Unit checks involve one module, while integration checks involve two adjacent modules. Each factor $\ell_a$ in~\eqref{eq:factors} indicates whether one visible or hidden check fails.

In T1--T6, either of two adjacent modules can perform the required interface transformation. We verified that assigning the transformation to either module gives a solution with $F=0$, whereas applying it in both modules produces the integration failures reported in Table~\ref{tab:tasks}. We also checked the initial failure counts and traced the module calls made by each check against its declared dependencies.

\paragraph{Role of hidden checks.}
Agents can run the visible checks during repair. The coordinator evaluates both visible and hidden checks when accepting edits and stops a successful run when all checks pass.

\paragraph{Agent assignments and information.}
The main comparison uses \texttt{claude-haiku-4-5}, with three replicates per task, 21 task--replicate units, and 261 agent calls. The replication uses \texttt{claude-sonnet-5}, with one replicate per task, seven units, and 74 calls. Each call starts with a fresh context and a separate copy of the current program. The instructions specify the task, the assigned module, and the available tests.

A module agent can inspect the starting program but may edit only its assigned module. Agents working in parallel receive the same starting program and cannot see one another's edits during the round. Each agent can run the visible unit and integration checks involving its module and is told that evaluation will also use hidden checks of the same kinds. Only changes to the assigned module enter the proposal used by the coordinator. The single-agent baseline receives the same task specification, with permission to edit all three modules and run all visible checks.

\paragraph{Repair schedules and stopping rules.}
All protocols stop when $F=0$ or when they reach their prescribed limit. \textsf{SGL} allows up to four whole-program repair sessions, each starting from the most recently accepted program. \textsf{SEQ} allows two passes through the modules in a fixed order and skips a module when all checks involving it already pass. Both baselines accept an edit only when it reduces $F$.

The four parallel protocols allow up to three rounds. Each round assigns one agent to every module involved in a failing check, with all agents starting from the same program. The protocols share their first-round proposals, allowing a paired comparison of their merge decisions, and each protocol is charged for the calls it uses. Later rounds generate fresh proposals from each protocol's own accepted program.

\paragraph{Merge decisions.}
The parallel protocols differ in how they combine those proposals. \textsf{MA} applies their union directly, while \textsf{MAG} accepts the union only when it reduces $F$. \textsf{COL} evaluates each changed module separately against the starting program, retains strictly improving proposals, and selects the color class with the largest summed gain, as in~\eqref{eq:fractional-merge-rule}. \textsf{OPT} also evaluates the union of the retained proposals and selects it when its gain exceeds that of the color candidate~\eqref{eq:optimistic}.

\paragraph{Reusing deferred proposals.}
After a color-class merge, verified proposals from the other class are retained and evaluated against the updated program before the next agent round. This permits further progress without requesting the same edit again. Together with skipping modules whose checks pass, reuse makes the assignment of work depend on the current program. The pointwise merge guarantee applies when the proposals satisfy Assumption~\ref{ass:structure}; the convergence rate in Theorem~\ref{thm:rate} additionally requires its fixed conditional assignment distribution and uniform local-improvement bounds. The reported costs include the savings from proposal reuse as well as the cost of verification.

\paragraph{Input cost and verification charges.}
We measure input-token-equivalent cost as input tokens $+\,1.25\times$ cache-write tokens $+\,0.1\times$ cache-read tokens. This measure excludes output tokens. One budget unit is the median call cost: 109{,}679 input-token equivalents for \texttt{claude-haiku-4-5} and 75{,}441 for \texttt{claude-sonnet-5}.

To account for verification, we add $v\in\{0,1/4,1\}$ budget units for each full-objective evaluation required by a protocol. This includes the initial evaluation and evaluations of individual proposals, unions, and candidates subject to an acceptance gate. The loss of a color candidate follows from additivity~\eqref{eq:fractional-add}, so the direct evaluations used only to check that identity are excluded from protocol cost.

\paragraph{Modeled completion time.}
We reconstruct completion time from measured agent-call durations, assuming that calls in a parallel round start together. A round therefore contributes the duration of its slowest call, followed by an assumed latency $t_{\mathrm{ver}}\in\{0,60,300\}$ seconds for each sequential verification step. Individual proposals are verified in parallel and together contribute one step; a subsequent union evaluation contributes another. The durations were recorded with up to 20 concurrent agents. Varying $v$ and $t_{\mathrm{ver}}$ shows how verification cost changes the comparison. Budget-success curves include every run, while median costs to solve are computed over successful runs.

\paragraph{Experimental design and model replication.}
The tasks, agent instructions, protocols, stopping limits, cost measures, and planned analyses were fixed before the first main-experiment call. After observing nearly identical first-round behavior across the \texttt{claude-haiku-4-5} replicates, we added a \texttt{claude-sonnet-5} replication to examine whether the conflict pattern depended on the model. The tasks, protocols, and cost accounting were unchanged.

\FloatBarrier
\section{Repair outcomes and coordination costs}\label{app:results}

Across the two models, both verified merge rules solve all 11 runs with diagnosed edit conflicts, while \textsf{MA} solves none and \textsf{MAG} solves one (Table~\ref{tab:repair-costs}). The following diagnostic identifies the ownership decisions that produce these differences.

\paragraph{Conflict diagnosis and model dependence.}
We diagnose an ownership conflict by testing the two adjacent proposals while using the reference implementation for all other modules. We classify the proposals as conflicting when each passes the interface check individually with the transformation absent from the adjacent module, but applying both proposals causes the check to fail. This diagnostic identified conflicts in 9 of the 18 \texttt{claude-haiku-4-5} units on T1--T6: all replicates of T1 and T6, two of T3, and one of T5. \texttt{claude-sonnet-5} produced conflicts in 2 of its 6 units, on T2 and T6. The models differed in how they assigned ownership: \texttt{claude-haiku-4-5} double-applied the coupon, while \texttt{claude-sonnet-5} double-escaped HTML. Across both models, every outcome difference among the parallel rules arose on a diagnosed conflict. All 108 audited color merges satisfied additivity~\eqref{eq:fractional-add} exactly.

\paragraph{Resolving a repeated ownership conflict.}
The first conflicted union lowers $F$ from 21 to 5 by repairing local defects despite applying the shared transformation twice.
In the next round, both owners respond to the failing integration checks by removing their copy of the transformation.
\textsf{MA} then alternates between both fixes and neither, yielding $21,5,5,5$.
\textsf{MAG} accepts the first union but rejects the union of both removals because it does not lower $F$.
The verified rules evaluate the removals separately and select one owner, giving the \textsf{OPT} trajectory $21,5,0$.
On T3, the conflicted union leaves five failures while the color candidate leaves four.
\textsf{OPT} keeps a union when it improves more than the color candidate and can resolve any remaining interference in a later round.
Table~\ref{tab:traj} gives all conflict trajectories.

\paragraph{Resource use and matched budgets.}
On these small tasks, \textsf{SGL} solves every unit with a median of one call. For \texttt{claude-haiku-4-5}, it uses about one quarter of the input-token-equivalent cost and one third of the modeled time of the verified teams at $t_{\mathrm{ver}}=0$. \textsf{SEQ} solves 26 of 28 units but has the longest modeled completion time. The verified teams typically need an additional round and roughly twice as many objective evaluations as merge-all, which can offset their higher eventual success rate under a fixed budget.

For example, at nine median-call-equivalent units and $v=1$, \textsf{MA} solves 12 of 21 \texttt{claude-haiku-4-5} units, while each verified rule solves eight. At a 15-minute budget with $t_{\mathrm{ver}}=300$ seconds, the corresponding counts are 12 and zero. Reducing verification latency to 60 seconds changes the comparison: within ten minutes, \textsf{COL}, \textsf{OPT}, and \textsf{MA} solve 17, 13, and 12 units, respectively. \texttt{claude-sonnet-5} shows the same qualitative ordering. Table~\ref{tab:budgets} and Figures~\ref{fig:budgets} and~\ref{fig:budgets-all} report these comparisons with unsuccessful runs included.

\begin{table}[h]
\caption{Success and resource use across the six repair protocols. The first four rows count solved task--replicate units; the conflict and other-unit rows pool both models. Resource-use rows are means over all 21 Haiku units. The final two rows give median cost among successful Haiku runs, in median-call-equivalent units, with each objective evaluation charged at $v$. Input costs exclude output tokens; ``k'' denotes thousands of input-token equivalents. Haiku denotes \texttt{claude-haiku-4-5}. Sonnet denotes \texttt{claude-sonnet-5}.}
\label{tab:repair-costs}
\centering
\small
\setlength{\tabcolsep}{2pt}
\begin{tabular}{lrrrrrr}
\toprule
& \textsf{SGL} & \textsf{SEQ} & \textsf{MA} & \textsf{MAG} & \textsf{COL} & \textsf{OPT}\\
\midrule
Solved, Haiku (21) & 21 & 19 & 12 & 13 & 21 & 21\\
Solved, Sonnet (7) & 7 & 7 & 5 & 5 & 7 & 7\\
Conflict units (11) & 11 & 11 & 0 & 1 & 11 & 11\\
Other units (17) & 17 & 15 & 17 & 17 & 17 & 17\\
\midrule
Agent calls & 1.14 & 3.14 & 4.71 & 4.71 & 3.86 & 3.86\\
Evaluations & 2.14 & 4.00 & 2.86 & 2.86 & 5.76 & 6.19\\
Input equivalents (k) & 137 & 341 & 633 & 640 & 567 & 559\\
Modeled time (s), $t_{\mathrm{ver}}=0$ & 85 & 272 & 263 & 269 & 267 & 250\\
Median cost to solve, $v=0$ & 1.06 & 3.04 & 3.36 & 3.47 & 4.41 & 4.41\\
Median cost to solve, $v=1$ & 3.06 & 7.04 & 5.36 & 5.47 & 9.41 & 9.41\\
\bottomrule
\end{tabular}
\end{table}

\begin{table}[h]
\caption{Solved runs and first-round edit conflicts for each repair task. Protocol columns count successful runs out of three replicates for Haiku and one for Sonnet. ``Conf.'' counts runs with a genuine conflict between adjacent proposals, as defined in the conflict diagnostic. C0 has no ambiguous interface. Haiku denotes \texttt{claude-haiku-4-5}. Sonnet denotes \texttt{claude-sonnet-5}.}
\label{tab:pertask}
\centering
\small
\begin{tabular}{l*{7}{c}@{\hspace{1.5em}}*{7}{c}}
\toprule
&\multicolumn{7}{c}{Haiku}&\multicolumn{7}{c}{Sonnet}\\
\cmidrule(lr){2-8}\cmidrule(lr){9-15}
Task & Conf. & \textsf{SGL} & \textsf{SEQ} & \textsf{MA} & \textsf{MAG} & \textsf{COL} & \textsf{OPT} & Conf. & \textsf{SGL} & \textsf{SEQ} & \textsf{MA} & \textsf{MAG} & \textsf{COL} & \textsf{OPT}\\
\midrule
T1 & 3 & 3 & 3 & 0 & 0 & 3 & 3 & 0 & 1 & 1 & 1 & 1 & 1 & 1\\
T2 & 0 & 3 & 2 & 3 & 3 & 3 & 3 & 1 & 1 & 1 & 0 & 0 & 1 & 1\\
T3 & 2 & 3 & 3 & 1 & 1 & 3 & 3 & 0 & 1 & 1 & 1 & 1 & 1 & 1\\
T4 & 0 & 3 & 3 & 3 & 3 & 3 & 3 & 0 & 1 & 1 & 1 & 1 & 1 & 1\\
T5 & 1 & 3 & 3 & 2 & 3 & 3 & 3 & 0 & 1 & 1 & 1 & 1 & 1 & 1\\
T6 & 3 & 3 & 3 & 0 & 0 & 3 & 3 & 1 & 1 & 1 & 0 & 0 & 1 & 1\\
C0 & --- & 3 & 2 & 3 & 3 & 3 & 3 & --- & 1 & 1 & 1 & 1 & 1 & 1\\
\bottomrule
\end{tabular}
\end{table}

\begin{table}[h]
\caption{Runs solved within fixed cost or time budgets, out of 21 Haiku runs and seven Sonnet runs. ``Calls'' denotes median-call-equivalent input cost, with an additional charge of $v$ per objective evaluation, rather than a limit on the number of calls. Time budgets include $t_{\mathrm{ver}}$ for each sequential verification step. All runs contribute to the denominators. Haiku denotes \texttt{claude-haiku-4-5}. Sonnet denotes \texttt{claude-sonnet-5}.}
\label{tab:budgets}
\centering
\small
\begin{tabular}{l*{6}{c}@{\hspace{1.5em}}*{6}{c}}
\toprule
&\multicolumn{6}{c}{Haiku (of 21)}&\multicolumn{6}{c}{Sonnet (of 7)}\\
\cmidrule(lr){2-7}\cmidrule(lr){8-13}
Budget & \textsf{SGL} & \textsf{SEQ} & \textsf{MA} & \textsf{MAG} & \textsf{COL} & \textsf{OPT} & \textsf{SGL} & \textsf{SEQ} & \textsf{MA} & \textsf{MAG} & \textsf{COL} & \textsf{OPT}\\
\midrule
3 calls, $v=0$ & 21 & 8 & 3 & 3 & 3 & 3 & 7 & 5 & 3 & 3 & 3 & 3\\
3 calls, $v=1$ & 9 & 0 & 0 & 0 & 0 & 0 & 3 & 0 & 0 & 0 & 0 & 0\\
6 calls, $v=\nicefrac14$ & 21 & 19 & 12 & 12 & 12 & 12 & 7 & 7 & 5 & 5 & 5 & 5\\
9 calls, $v=1$ & 21 & 19 & 12 & 12 & 8 & 8 & 7 & 7 & 5 & 5 & 4 & 4\\
5 min, $t_{\mathrm{ver}}=60$ s & 18 & 0 & 11 & 11 & 3 & 3 & 7 & 0 & 5 & 5 & 5 & 5\\
10 min, $t_{\mathrm{ver}}=60$ s & 21 & 17 & 12 & 12 & 17 & 13 & 7 & 7 & 5 & 5 & 7 & 7\\
15 min, $t_{\mathrm{ver}}=300$ s & 18 & 0 & 12 & 12 & 0 & 0 & 7 & 0 & 5 & 5 & 0 & 0\\
\bottomrule
\end{tabular}
\end{table}

\begin{figure}[h]
\centering
\includegraphics[width=\textwidth]{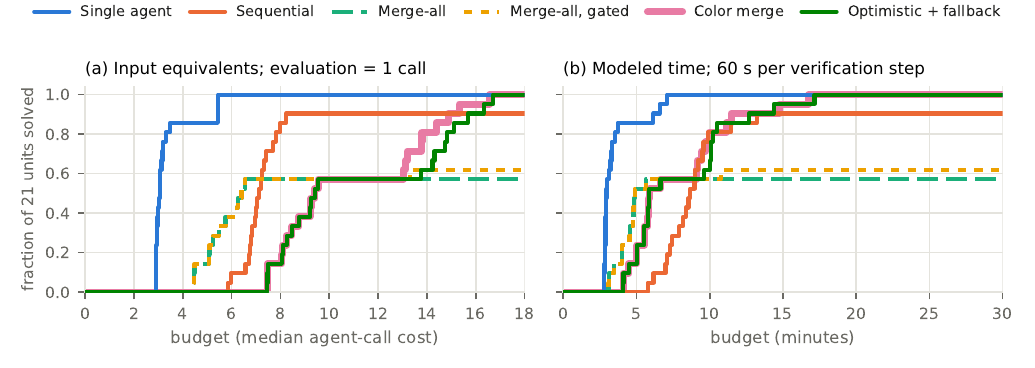}
\caption{Fraction of the 21 Haiku runs solved as the available budget increases. (a) Input cost, measured in median-call equivalents and excluding output tokens, with one additional unit per objective evaluation ($v=1$). (b) Modeled completion time, combining the measured agent critical path with 60 seconds per sequential verification step. Unsuccessful runs remain in the denominator. Figure~\ref{fig:budgets-all} varies the verification charges. Haiku denotes \texttt{claude-haiku-4-5}.}
\Description{Success rates of the six repair protocols as the input-token-equivalent and modeled completion-time budgets increase.}
\label{fig:budgets}
\end{figure}

\begin{figure}[h]
\centering
\includegraphics[width=\textwidth]{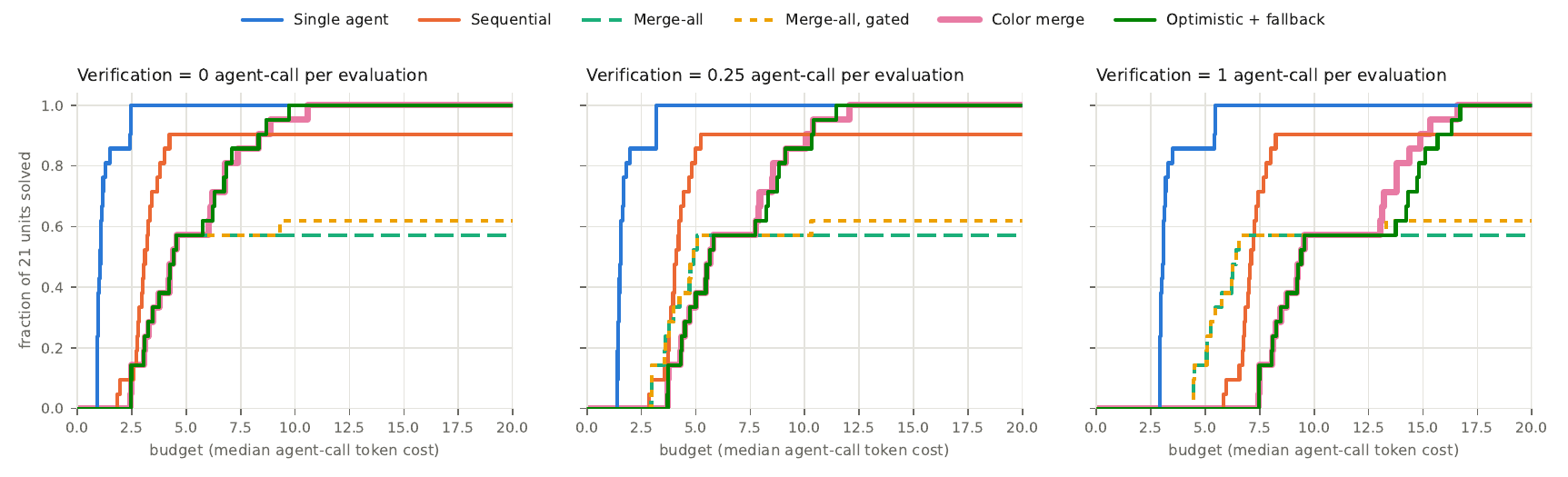}\\[2mm]
\includegraphics[width=\textwidth]{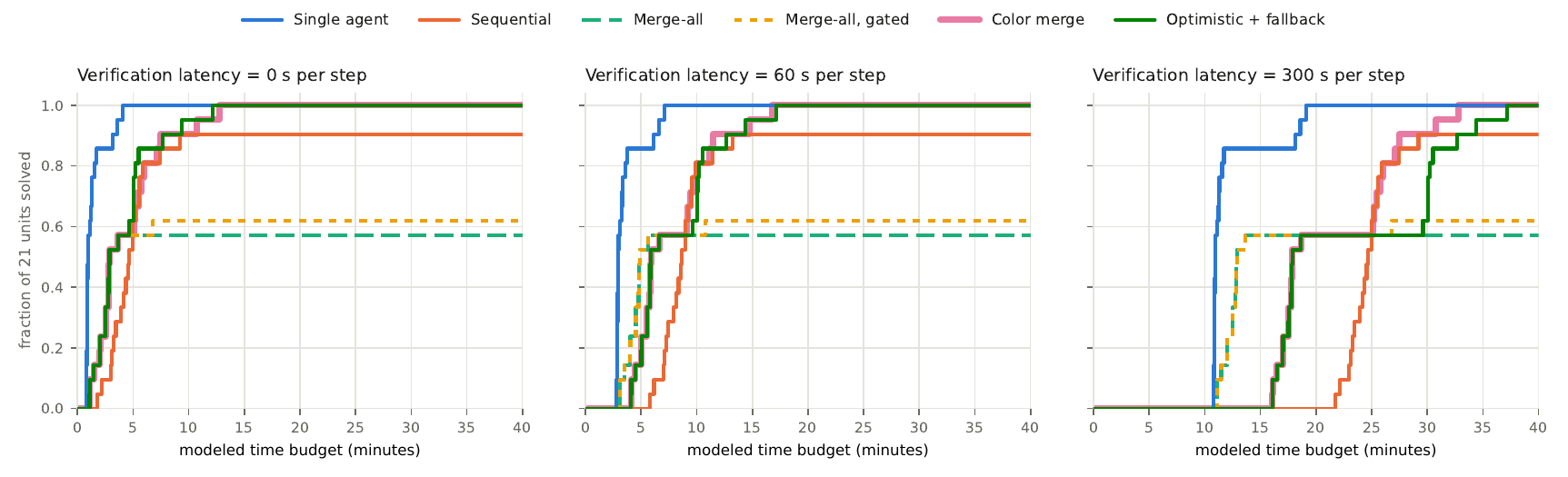}
\caption{Effect of verification charges on the fraction of Haiku runs solved within a budget. The top row varies the price of one objective evaluation over $v=0,\nicefrac14,1$ median-call equivalents. The bottom row varies sequential verification latency over $t_{\mathrm{ver}}=0,60,300$ seconds. Every curve uses all 21 runs, including failures; input costs exclude output tokens. Haiku denotes \texttt{claude-haiku-4-5}.}
\Description{Success-rate curves for six repair protocols under three token-accounting prices and three verification-time assumptions.}
\label{fig:budgets-all}
\end{figure}

\begin{table}[h]
\caption{Success and resource use in the seven-run Sonnet replication. ``Solved'' counts successful runs. Agent calls, objective evaluations, input-token equivalents (thousands), and modeled time are means over all seven runs. The final two rows report median cost among successful runs, in median-call-equivalent units, charging $v$ per objective evaluation as in Table~\ref{tab:repair-costs}. Sonnet denotes \texttt{claude-sonnet-5}.}
\label{tab:sonnet}
\centering
\small
\begin{tabular}{lrrrrrr}
\toprule
& \textsf{SGL} & \textsf{SEQ} & \textsf{MA} & \textsf{MAG} & \textsf{COL} & \textsf{OPT}\\
\midrule
Solved & 7 & 7 & 5 & 5 & 7 & 7\\
Agent calls & 1.00 & 3.14 & 4.14 & 4.14 & 3.57 & 3.57\\
Evaluations & 2.00 & 4.14 & 2.57 & 2.57 & 5.57 & 5.86\\
Input eq. (k) & 78 & 216 & 312 & 323 & 277 & 279\\
Modeled time (s), $t_{\mathrm{ver}}=0$ & 38 & 99 & 69 & 74 & 73 & 65\\
Cost to solve, $v=0$ & 1.05 & 2.82 & 2.89 & 2.89 & 3.14 & 3.14\\
Cost to solve, $v=1$ & 3.05 & 6.82 & 4.89 & 4.89 & 8.14 & 8.14\\
\bottomrule
\end{tabular}
\end{table}

\begin{table}[h]
\caption{Failing checks after successive merges on the nine Haiku runs with edit conflicts. Each sequence begins with the initial loss $F(x^0)$ and ends at success or the round cap. Runs grouped in the same row have identical trajectories. Haiku denotes \texttt{claude-haiku-4-5}.}
\label{tab:traj}
\centering
\small
\begin{tabular}{lllll}
\toprule
Task and replicates & \textsf{MA} & \textsf{MAG} & \textsf{COL} & \textsf{OPT}\\
\midrule
T1, replicates 1--3 & 21, 5, 5, 5 & 21, 5, 5, 5 & 21, 6, 5, 0 & 21, 5, 0\\
T3, replicates 2--3 & 26, 5, 6, 5 & 26, 5, 5, 5 & 26, 4, 4, 0 & 26, 4, 4, 0\\
T5, replicate 3 & 25, 5, 5, 5 & 25, 5, 5, 0 & 25, 14, 5, 0 & 25, 5, 0\\
T6, replicates 1--3 & 26, 6, 6, 6 & 26, 6, 6, 6 & 26, 14, 5, 0 & 26, 6, 0\\
\bottomrule
\end{tabular}
\end{table}

\paragraph{First merges and the interaction surplus.}
When ownership is consistent, repairs in adjacent modules can reinforce each other by making their shared integration checks pass. The first union solved all 12 conflict-free \texttt{claude-haiku-4-5} units and all 5 conflict-free \texttt{claude-sonnet-5} units. Merge-all and \textsf{OPT} therefore finished in one round, while the color candidate left a mean of 12.3 and 11.6 failing checks, respectively. \textsf{COL} needed a further round. On the 9 \texttt{claude-haiku-4-5} units with conflicts, the first merge left a mean of 5.3 failures for merge-all, 9.1 for \textsf{COL}, and 5.1 for \textsf{OPT}.

Across both experiments, \textsf{OPT} evaluated 37 unions, accepted 26, and selected the color candidate 11 times. Two fallbacks occurred in the first round on T3 conflict units; the others occurred later, when the union offered no improvement over a single color class. The trajectories in Table~\ref{tab:traj} show how verified merging resolves these conflicts, while merge-all retains 5--6 failures through the round cap. \textsf{MAG} escaped only on T5 replicate 3: after its second-round union was rejected, fresh agents in the third round removed only one copy of the transformation.

\paragraph{Sequential repair.}
Conflicts can also arise when edits are applied sequentially. In the \texttt{claude-sonnet-5} run on T4, the records agent added quoting in the first pass, and the csvrow agent then quoted the fields again. The acceptance gate kept the second edit because its local repairs reduced $F$ despite the double quoting. In the second pass, the records agent removed its quoting and solved the task. The two \texttt{claude-haiku-4-5} failures, T2 replicate 1 and C0 replicate 1, reached the pass cap with one hidden check still failing: name trimming and total precision, respectively.

\paragraph{Replicate dependence and task-level comparisons.}
The seven tasks and three modules define 21 distinct (task, module) groups. Pass/fail behavior varied across replicates only for T3 timeparse and T5 slugs, which account for the reported differences in conflict incidence. This strong dependence makes the task the more informative unit of comparison. Verified merging outperforms merge-all on T1, T3, T5, and T6 and ties on the other tasks. The four wins and no losses give a two-sided sign-test $p$-value of $2\cdot2^{-4}=0.125$; treating the 9 conflict units as independent would instead give $p\approx0.004$. The \texttt{claude-sonnet-5} replication changes which tasks produce conflicts while preserving the observed mechanism and ordering of the rules.

\FloatBarrier
\section{Scaling to larger repair tasks}\label{app:scale}
The three-module benchmark contains programs that a single agent can repair readily. The scaling benchmark tests when concurrent repair offsets coordination overhead. Each task contains $m$ modules with five bugs per module. Function contracts specify ownership explicitly, and agents are instructed to rely on the contracts of other modules. Checks involve at most one base module and one composite module, yielding a nonempty bipartite dependency graph with $\chi_f(G)=\chi(G)=2$. Each bug causes a visible check to fail both individually and alongside the other bugs in its module. Agents perform their assigned work without delegation.

\begin{table}[h]
\caption{Size and composition of the twelve scaling tasks. Each size $m$ has three instances with different bug seeds and five bugs per module. Clusters A--H cover commerce, logistics, scheduling, analytics, inventory, geometry, identifiers, and text analytics; each contains two base and two composite modules. Instance names encode the number of modules; suffixes a--c distinguish the three instances at each size. Multiple entries in a cell follow the listed instance order. $F(x^0)$ counts initially failing checks.}
\label{tab:scale-tasks}
\centering
\small
\begin{tabular}{llrrrr}
\toprule
Instances & Clusters & $m$ & Bugs & Checks (visible/hidden) & $F(x^0)$\\
\midrule
S4a, S4b, S4c & A; C; D & 4 & 20 & 36/36, 37/37, 36/36 & 47, 45, 48\\
S8a, S8b, S8c & A+B; C+D; B+D & 8 & 40 & 72/72, 73/73, 72/72 & 87, 105, 95\\
S16a, S16b, S16c & A--D & 16 & 80 & 145/145 & 206, 190, 182\\
L32a, L32b, L32c & A--H & 32 & 160 & 289/289 & 386, 383, 363\\
\bottomrule
\end{tabular}
\end{table}

\paragraph{Task construction.}
The benchmark uses 32 Python modules grouped into the eight domains in Table~\ref{tab:scale-tasks}. Each function has a documented contract and visible and hidden input--output checks generated from the reference implementation. A composite function calls at most one base module, so a check's declared dependencies consist of its own module and the base module it calls. Tasks with $m=4,8,16,32$ contain 72--74, 144--146, 290, and 578 checks, respectively.

We introduce bugs by changing operators or constants in the syntax tree while preserving the surrounding source: swapping relational or arithmetic operators, exchanging logical conjunction and disjunction, changing an integer by $\pm1$, exchanging minimum and maximum, or removing a negation. Each instance uses a separate bug seed. We retain a mutation only when the affected function's visible check fails both with that mutation alone and with the module's other mutations, and the program still terminates. Before running agents, we verified the initial failure counts, checked that the reference implementation passes every test, and traced the dependencies exercised by the checks. The base and composite modules form a proper two-coloring, and reference repairs have additive gains within each color class.

\paragraph{Agent instructions and execution.}
All 278 calls use \texttt{claude-haiku-4-5}: 12 whole-program calls, 84 sequential module calls, and 182 parallel module calls. Agents receive separate workspaces and fresh contexts as in Appendix~\ref{app:bench}. Module agents are instructed to repair their own functions according to the specification and assume that other modules satisfy their documented contracts, even when an integration check currently fails. The instructions therefore direct agents to repair their assigned module rather than compensate for defects in neighboring modules. The single agent can choose which modules to test.

The recorded calls used no additional agents and remained within their assigned workspaces and editable modules. We saved the single agent's program every five seconds to measure repair progress.

\paragraph{Protocols and accounting.}
The protocols, iteration limits, input-cost measure, and time model follow Appendix~\ref{app:bench}. We omitted \textsf{SEQ} at $m=32$ because each pass requires 32 sequential calls and it was already the slowest protocol at $m\le16$. Parallel rules share their first-round calls. In later rounds, they also share a call when they request the same module from an identical artifact. Each rule is charged for every call it uses. The median call cost was 123{,}539 input-token equivalents.

Calls were launched in batches of at most 20, while the time model assumes simultaneous starts within a round. To assess whether the worker limit would change the modeled duration, we scheduled the recorded calls on 20 workers in descending order of duration. Every round still finished within the duration of its slowest call. Thus the recorded calls fit within the assumed round duration on 20 workers. The modeled times omit queueing and orchestration delays.

\begin{table}[h]
\caption{Time and input cost to solve tasks with $m$ modules, averaged geometrically over three instances at each size. The first two column groups give modeled minutes at verification latencies of 0 and 60 seconds per sequential step. The final group gives millions of input-token equivalents, excluding output tokens and verification charges ($v=0$). \textsf{MAG} matches \textsf{MA}; at zero verification latency, \textsf{COL} matches \textsf{OPT} because its extra verification step needs no agent call. Dashes mark the omitted 32-module \textsf{SEQ} runs.}
\label{tab:scale-full}
\centering
\small
\begin{tabular}{r*{5}{c}@{\hspace{1.2em}}*{5}{c}@{\hspace{1.2em}}ccc}
\toprule
&\multicolumn{5}{c}{Minutes, $t_{\mathrm{ver}}=0$}&\multicolumn{5}{c}{Minutes, $t_{\mathrm{ver}}=60$ s}&\multicolumn{3}{c}{Input equivalents (M)}\\
\cmidrule(lr){2-6}\cmidrule(lr){7-11}\cmidrule(lr){12-14}
$m$ & \textsf{SGL} & \textsf{SEQ} & \textsf{MA} & \textsf{COL} & \textsf{OPT} & \textsf{SGL} & \textsf{SEQ} & \textsf{MA} & \textsf{COL} & \textsf{OPT} & \textsf{SGL} & \textsf{SEQ} & \textsf{OPT}\\
\midrule
4 & 2.4 & 4.8 & 3.5 & 3.5 & 3.5 & 4.4 & 9.9 & 5.6 & 6.6 & 6.6 & 0.25 & 0.47 & 0.56\\
8 & 4.2 & 8.6 & 2.9 & 2.9 & 2.9 & 6.2 & 17.6 & 5.0 & 6.0 & 6.0 & 0.49 & 0.86 & 1.09\\
16 & 9.2 & 19.3 & 4.5 & 4.5 & 4.5 & 11.2 & 36.3 & 6.8 & 7.8 & 7.8 & 1.19 & 1.80 & 2.30\\
32 & 15.3 & --- & 4.1 & 4.1 & 4.1 & 17.3 & --- & 6.1 & 7.1 & 7.1 & 2.08 & --- & 4.22\\
\bottomrule
\end{tabular}
\end{table}

\begin{figure}[h]
\centering
\includegraphics[width=\textwidth]{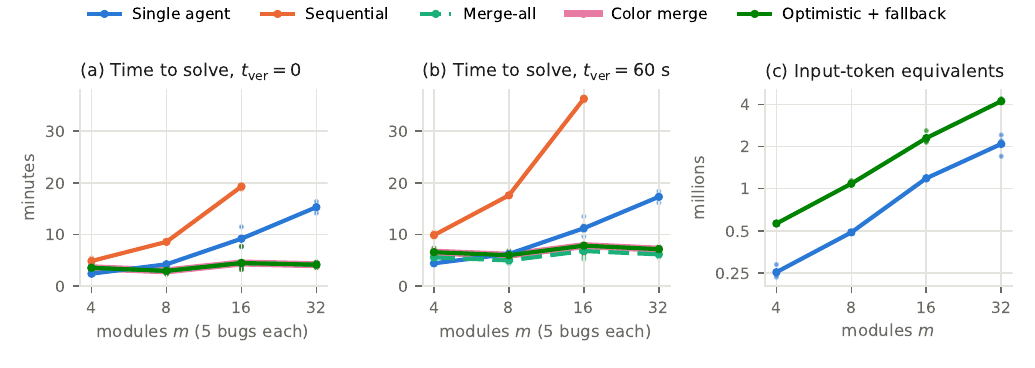}
\caption{Time and input cost to solve increasingly large repair tasks. Panels (a) and (b) show modeled completion time with 0 and 60 seconds per sequential verification step, respectively; panel (c) shows input-token equivalents, excluding output tokens. Lines connect geometric means over three instances per size, and dots show the individual instances. In panel (a), \textsf{MA}, \textsf{COL}, and \textsf{OPT} have equal modeled times.}
\Description{Modeled completion time and input-token-equivalent cost of the repair protocols as the number of modules increases.}
\label{fig:scale}
\end{figure}

\paragraph{Outcomes and first merges.}
Every protocol passed all visible and hidden checks on every instance it ran. \textsf{SGL} completed each task in one session, and \textsf{SEQ} completed each in one pass. \textsf{MA}, \textsf{MAG}, and \textsf{OPT} usually finished in one round. Their exception was the third 16-module instance (S16c in Table~\ref{tab:scale-tasks}), with trajectory $182,2,0$: the invoice agent left one function unrepaired, and the same two checks failed when its proposal was tested alone against the reference. \textsf{COL} completed repair after re-verifying its deferred class, as in the trajectory $386,153,0$.

All $m$ first-round proposals improved $F$ individually. Their union outperformed the color candidate on all 12 instances, leaving zero failures except for the two on that 16-module instance, compared with 19--156 failures for the color candidate. The positive interaction surplus~\eqref{eq:surplus} ranged from 5 to 36 checks: a check involving both a faulty base function and a faulty composite function passes only after both are repaired. Additivity~\eqref{eq:fractional-add} held in all 38 audited color merges, and \textsf{OPT} always selected the union over the verified color candidate.

\paragraph{Sources of the modeled time advantage.}
The timing results in Figure~\ref{fig:scale} reflect a balance between the single agent's efficiency and the duration of the slowest parallel call. The single agent spent 26.5--43.0 seconds per module, with a geometric mean of 28.6 seconds at $m=32$, and ran tests 3--13 times per session. Parallel calls repaired only five bugs each but took a median of 97 seconds, used 16 tools, and consumed 128k input-token equivalents. Median call duration varied from 80 to 152 seconds across instances. Within a round, the slowest call took 1.45--3.5 times the median, setting the modeled agent time at 140--280 seconds.

As the number of modules grows, the single agent's roughly linear repair time eventually exceeds this parallel-round duration. At $t_{\mathrm{ver}}=0$, the crossover occurs between $m=4$ and $m=8$; at $t_{\mathrm{ver}}=60$ seconds, the speedup at $m=8$ is only 1.04. Verification also separates the parallel rules: \textsf{OPT} evaluates individual proposals and then their union, while \textsf{MA} evaluates only the union. Consequently, \textsf{MA} has a shorter modeled completion time when verification latency is positive.

The time saving comes with greater input cost. The single agent's peak context averaged 73k, 94k, 115k, and 122k tokens at $m=4,8,16,32$, compared with a median peak of 64k for parallel calls. Cache reuse keeps the single agent's input-token-equivalent cost approximately linear over these sizes, whereas each module agent must read the task and program in a fresh context. The team therefore costs about twice as much at every tested size. Scaling beyond these instances will depend on the available concurrency and the growth of the maximum call duration.

\begin{figure}[h]
\centering
\includegraphics[width=\textwidth]{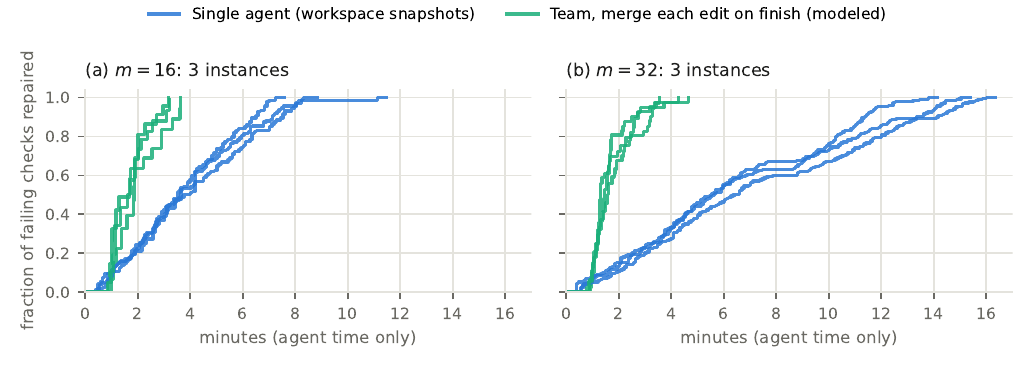}
\caption{Fraction of initially failing checks repaired over time for tasks with (a) 16 and (b) 32 modules. Single-agent curves evaluate workspace snapshots taken every five seconds on visible and hidden checks. Team curves model simultaneous first-round starts and immediate merging as each call finishes, with no verification delay. Each curve represents one instance; the team curves provide an optimistic merge-on-finish reference for comparison with the single-agent runs.}
\Description{Anytime repair trajectories for a single agent and a parallel team on tasks with 16 and 32 modules.}
\label{fig:scale-anytime}
\end{figure}

\paragraph{Anytime progress.}
Figure~\ref{fig:scale-anytime} and Table~\ref{tab:scale-anytime} show how quickly repairs accumulate. At $m=32$, the modeled merge-on-finish team reaches 90\% repaired 4.61$\times$ sooner than the single agent, while the slowest call reduces the completion ratio to 3.61$\times$. These curves provide an optimistic reference by merging edits as they arrive and omitting verification. The executed \textsf{OPT} protocol waits for the round's proposals and verifies them before merging, so its intermediate progress follows a different schedule.

\begin{table}[h]
\caption{Relative time to reach 50\%, 90\%, and 100\% repair under the merge-on-finish model. Entries are geometric means over three instances of single-agent time divided by modeled team time; values above one favor the team. Repair fractions refer to initially failing checks. The completion ratio at $m=16$ is omitted because one first-round union left failures.}
\label{tab:scale-anytime}
\centering
\small
\begin{tabular}{rccc}
\toprule
$m$ & 50\% & 90\% & 100\%\\
\midrule
4 & 0.79 & 0.53 & 0.57\\
8 & 1.38 & 1.26 & 1.32\\
16 & 2.13 & 2.26 & ---\\
32 & 4.06 & 4.61 & 3.61\\
\bottomrule
\end{tabular}
\end{table}

\paragraph{Sensitivity to slow tests.}
Test runs here take seconds. To see how slow tests would shift the comparison, we added a modeled cost to every recorded test-runner invocation and kept each agent's measured behavior fixed (Table~\ref{tab:scale-slow}). A module agent runs only the 9--18 visible checks that involve its module. The single agent ran all visible checks on every invocation in 10 of 12 instances. A cost per executed check therefore favors the team: at $m=32$ and one second per check, \textsf{OPT} would be 7.41$\times$ faster at $t_{\mathrm{ver}}=0$ and 4.92$\times$ at $t_{\mathrm{ver}}=60$ seconds. A fixed cost per invocation favors the single agent, which invoked the runner about as often as each module agent. The projected ratios reflect both parallel repair and selective testing, with the recorded testing behavior held fixed.

\begin{table}[h]
\caption{Projected speedup of \textsf{OPT} over \textsf{SGL} when tests take longer. Entries are geometric means of single-agent time divided by team time, adding the stated latency per visible check or test-runner invocation while holding recorded agent behavior fixed. Columns vary task size and verification latency. The ``None'' row removes measured test time, explaining the small difference from Table~\ref{tab:scale}.}
\label{tab:scale-slow}
\centering
\small
\begin{tabular}{lcccc}
\toprule
&\multicolumn{2}{c}{$m=16$}&\multicolumn{2}{c}{$m=32$}\\
\cmidrule(lr){2-3}\cmidrule(lr){4-5}
Added test cost & $t_{\mathrm{ver}}=0$ & $t_{\mathrm{ver}}=60$ s & $t_{\mathrm{ver}}=0$ & $t_{\mathrm{ver}}=60$ s\\
\midrule
None & 2.09 & 1.44 & 3.73 & 2.42\\
0.1 s per check & 2.26 & 1.55 & 4.24 & 2.73\\
0.5 s per check & 2.91 & 1.95 & 5.86 & 3.79\\
1 s per check & 3.60 & 2.42 & 7.41 & 4.92\\
10 s per invocation & 1.86 & 1.37 & 3.35 & 2.33\\
30 s per invocation & 1.50 & 1.22 & 2.80 & 2.14\\
60 s per invocation & 1.18 & 1.05 & 2.26 & 1.88\\
\bottomrule
\end{tabular}
\end{table}

\FloatBarrier
\section{Comparing ordinary and fractional coloring}\label{app:fractional}

\paragraph{Units and proposal generation.}
The preceding benchmarks have bipartite dependency graphs, for which ordinary and fractional coloring both give $\chi_f(G)=\chi(G)=2$. To separate their guarantees, we use the five-cycle $C_5$, where $\chi(C_5)=3$ and $\chi_f(C_5)=5/2$. The benchmark contains twelve explicit-contract units spanning four domains and three ordinary-coloring rotations fixed before the runs. Five fresh \texttt{gpt-6-luna} module agents generate proposals for each unit in a single run, and both coloring rules receive the same first-round proposals.

\paragraph{Coloring rules and retained gain.}
Let $S$ be the sum of positive verified singleton gains. \textsf{COL3} selects the best class of the fixed ordinary coloring, guaranteeing a gain of at least $S/3$. \textsf{FRAC5} evaluates the five independent pairs of $C_5$, raising the guarantee to $2S/5$. These pairs are all the maximal independent sets of $C_5$, so with nonnegative gains the best pair is also a maximum-weight independent set.

On the paired proposals, \textsf{FRAC5} matched or exceeded \textsf{COL3} in every unit and improved on it in four of the twelve, raising the mean retained fraction from 0.5615 to 0.6167. Every additivity audit passed. Figure~\ref{fig:graph-diagnostic} shows these paired fractions alongside the directly evaluated full-union gains. Using the same proposals for both rules isolates the effect of the coloring choice.

\FloatBarrier
\section{Held-out interface migrations}\label{app:interface}

\paragraph{Task construction.}
The benchmark contains sixteen held-out interface-migration designs, eight with five modules on $C_5$ and eight with ten modules on a prism graph.
The designs vary whether interface ownership is shared or explicit and whether the initial program contains an additional local bug.
The required transformations are non-idempotent: applying a transformation twice can introduce a failure even when either application alone is correct.
Each design has three replicates, giving 48 task--replicate units grouped into sixteen designs.
The objective counts failed visible and hidden checks, and we audit each check's module dependencies against the declared graph.

\paragraph{Agent instructions and repair protocol.}
All methods use the same version of \texttt{Qwen2.5-Coder-32B-Instruct}.
Module agents receive the specification, current program, and visible-check feedback, with edits restricted to their assigned module.
The coordinator uses both visible and hidden checks to evaluate proposals.
We compare \textsf{MA}, \textsf{MAG}, \textsf{COL}, and \textsf{OPT} with safe fractional integration (\textsf{FRAC}) and its optimistic variant (\textsf{FOPT}).
\textsf{FRAC} selects the safe fractional-coloring candidate; \textsf{FOPT} also evaluates the full union and retains the better candidate.
All six methods start from the same faulty program and reuse the same first-round proposals.
Each method then requests proposals from its own accepted program, stopping at zero loss or after three rounds in total.
This gives paired first-round comparisons and allows later proposals to adapt to each method's repair trajectory.

\paragraph{Repair outcomes.}
Table~\ref{tab:interface-heldout} reports the final losses.
\textsf{FOPT} leaves fewer failed checks than \textsf{MA} and \textsf{MAG} on all 48 units.
The first-round fractional candidate improves on the full union in 47 units; \textsf{FOPT} selects the union in the remaining unit.
All safe-candidate additivity audits pass.
None of the six parallel methods reaches zero loss within three rounds.

The fractional candidate retains more gain than the ordinary-color candidate on 30 of the 48 shared first-round proposal sets.
After later rounds, \textsf{FRAC} finishes with lower loss than \textsf{COL} in ten units, equal loss in 33, and higher loss in five.
\textsf{FOPT} and \textsf{FRAC} finish with equal loss in 47 units, while \textsf{FRAC} achieves lower loss in one.
Because later proposals are generated from different accepted programs, these final outcomes depend on the full repair trajectory.

\begin{table}[h]
\caption{Mean final failed-check counts on the held-out interface-migration benchmark. Results average over 48 task--replicate units overall and 24 per graph family; lower is better. Parallel rules run for at most three rounds. The multi-file control has a token cap of twice the corresponding \textsf{FOPT} run.}
\label{tab:interface-heldout}
\centering
\small
\begin{tabular}{lrrr}
\toprule
Method & Overall & $C_5$ & Prism \\
\midrule
Initial & 34.38 & 19.50 & 49.25 \\
\textsf{SGL} & 12.27 & 5.92 & 18.63 \\
\textsf{SGL-MF} ($\leq 2\times$) & 20.52 & 6.25 & 34.79 \\
\textsf{MA} & 20.02 & 10.25 & 29.79 \\
\textsf{MAG} & 19.58 & 10.25 & 28.92 \\
\textsf{COL} & 6.19 & 2.50 & 9.88 \\
\textsf{FRAC} & 6.17 & 2.50 & 9.83 \\
\textsf{OPT} & 7.92 & 4.50 & 11.33 \\
\textsf{FOPT} & 6.29 & 2.50 & 10.08 \\
\bottomrule
\end{tabular}
\end{table}

\paragraph{Single-agent controls.}
Both controls use the same model, designs, and replicate assignments as the parallel methods.
Each call starts with fresh context and receives the full current program and its aggregate loss.
The one-module control, \textsf{SGL}, chooses one module to edit per call and retains the edit only when the loss decreases.
It is allowed up to $3m$ calls, where $m$ is the number of modules: 15 calls for $C_5$ and 30 for the prism graph.
Its mean final loss is 12.27, and it solves no unit.
The same mean is reached when each run is truncated to the number of calls charged to \textsf{FOPT}.

The multi-file control, \textsf{SGL-MF}, may return complete source for any subset of the named modules in each call.
It also retains only improving edits and has a $3m$ call limit, together with a token cap of twice the corresponding \textsf{FOPT} run.
At the checkpoint matching \textsf{FOPT}'s token use, its mean loss is 24.17 and exceeds \textsf{FOPT}'s loss in all 48 units.
At twice that token use, its mean loss is 20.52: it finishes below \textsf{FOPT} in four units, ties in three, and finishes above it in 41.
It fully solves one unit.
Across these runs, 130 calls produce outputs that cannot be parsed, limiting the control's ability to apply its proposed repairs.

\paragraph{Resource accounting.}
\textsf{MA} uses a mean of 22.50 agent calls and 3.50 objective evaluations per unit, compared with 16.77 calls and 21.00 evaluations for \textsf{FOPT}.
Each method is charged for the calls it uses, including calls shared with other methods.
\textsf{SGL-MF} uses a mean of 96.62k total tokens, compared with 50.75k for \textsf{FOPT}.
Matching call counts or token use leaves differences in context, evaluation cost, and edit permissions; the parsing failures further affect the multi-file comparison.
The controls and parallel methods were executed separately, and shared calls prevent the combined parallel execution time from measuring each method's latency.
We therefore compare repair quality and charged resources on this benchmark.

\FloatBarrier
\section{Exploratory real-repository feasibility check}\label{app:external}

We evaluate module-restricted repair on two historical multi-fault versions of Python repositories from BugsInPy-MF~\cite{callaghan2025mining}: HTTPie version 2 and TheFuck version 2.
The cases and public tests were selected before agent inference.
Each case has three unseeded repeats, which are not three independent repository tasks.
The model is \texttt{gpt-6-sol} through a tool-using CLI, rather than the \texttt{Qwen2.5-Coder-32B-Instruct} model used in the held-out interface benchmark.
Agents can inspect source, edit allowed files, and run the published public tests.
Agents receive neither reference patches nor hidden tests.

In the module-restricted condition, each invocation may edit one source module. Proposals from the same graph-independent block are generated against the same accepted source, and their union is kept only if it strictly improves the selected fail-to-pass tests while preserving the selected pass-to-pass tests.
The global agent may edit any of the allowed source modules using the same model and public-test feedback. Its invocation limit equals the observed module-agent count for that repeat.
It serves as a cost reference when one agent may edit all allowed modules.
The graph is inferred from static imports, not certified against semantic interactions.
Both graphs satisfy $\chi_f(G)=\chi(G)=2$; graph blocks were run sequentially, so neither a fractional-coloring nor an observed parallel-time advantage is tested here.

\begin{table}[h]
\caption{Selected public-test outcomes on the real-repository check. Final counts hold separately for all three repeats of each repository in both the module-restricted and global conditions. F2P and P2P denote fail-to-pass and pass-to-pass tests.}
\label{tab:external-repair}
\centering
\small
\begin{tabular}{lccc}
\toprule
Repository & Base F2P & Final F2P & Final P2P \\
\midrule
HTTPie & $0/12$ & $12/12$ & $2/2$ \\
TheFuck & $2/4$ & $4/4$ & $2/2$ \\
\bottomrule
\end{tabular}
\end{table}

The related public test files also passed in every final snapshot (HTTPie $26/26$, TheFuck $34/34$), but these files are not whole-repository test suites.
Across the six paired repeats, the module-restricted condition used 12 agent invocations, 799.8 seconds of summed agent wall time, 2,491,797 input tokens (278,933 uncached), 29,604 output tokens, and 105 agent-initiated test commands.
The global condition used six invocations, 456.4 seconds, 1,470,456 input tokens (148,472 uncached), 17,759 output tokens, and 54 test commands.
These runs demonstrate that module-restricted agents can pass the selected public tests in two real repositories.

\FloatBarrier
\section{External exact-factor road-network merger}\label{app:road-merger}

\paragraph{Road networks and selected regions.}
The public SNAP roadNet-PA, roadNet-CA, and roadNet-TX networks~\cite{snap-roadnet} represent undirected road connections among intersections and endpoints. We evaluate proposal integration in 30 held-out regions, ten per network. Each region contains five connected blocks of 32 vertices, with the interblock roads forming a $C_5$ dependency graph. Three initializations per region give 90 paired comparisons.

\paragraph{Objective, agents, and common incumbent.}
Each road edge touching a selected block is one binary cut factor; edges outside the selected vertices contribute only constants. Each of five separate agents controls one block and attempts 30,000 within-block one-for-zero label swaps, preserving balance. A proposal is evaluated against a common incumbent; a nonpositive, invalid, or failed singleton is retained as null. After a development-only diagnostic, we froze a two-stage design before inspecting held-out outcomes. Stage one applies a guarded full union from the initial labels to create the same improved incumbent for every rule. This shared step gained 8,395 cut edges over the 90 units and is charged to every method. Stage two runs five \emph{new} agents from that incumbent.

\paragraph{Same-pool integration rules.}
Let $g_i\geq0$ be the verified singleton reduction in cut edges for block $i$, and let $S=\sum_i g_i$. The block conflict graph has an edge precisely when a road joins the blocks, so independent sets have no shared cut factors and their gains add exactly. \textsf{COL} selects the best pre-fixed class among $\{0,2\}$, $\{1,3\}$, and $\{4\}$. \textsf{FRAC} selects the best of the five nonadjacent pairs supporting the optimal $C_5$ fractional cover, each assigned weight $1/2$. \textsf{MA} applies all nonnull edits without a guard; \textsf{MAG} applies their union only if it improves the objective; \textsf{FOPT} takes the better of the fractional candidate and full union. The \emph{best singleton} rule selects the highest-gain proposal. \emph{Adaptive ordinary coloring} selects a coloring containing the best independent set after observing gains. These provide two further controls. Every rule scores the same stage-two proposal pool.

\paragraph{Retained gain.}
Table~\ref{tab:road-merger} gives per-network gain sums. Across all 90 units, $S=1{,}617$, \textsf{FRAC} retains 922 (57.0\%) and fixed \textsf{COL} 799 (49.4\%); the relative difference is 15.4\%. The best singleton retains 629, \textsf{MAG} 242, and \textsf{FOPT} 926. \textsf{FRAC} exceeds \textsf{MAG} in 84 units, loses in three, and ties in three. Unguarded \textsf{MA} worsens the incumbent in 31 units. Adaptive ordinary coloring equals \textsf{FRAC} in \emph{every} unit, as it can place the highest-gain independent pair in a proper three-coloring after observing gains.

\paragraph{Measured resources and interpretation.}
Both stages together used 900 process-level agent calls, 66.19 summed agent CPU-seconds, 14.15 seconds summed five-agent critical-path wall time, and 3,780 direct objective evaluations including runner audit checks. The complete local run, including graph loading and serialization, took 25.99 seconds on macOS arm64. Fixed \textsf{COL} scored three classes and \textsf{FRAC} five pairs from cached singleton gains; the latter required two extra score checks per unit but no additional agents or direct objective calls.

On these road networks, fractional integration retains more verified gain than ordinary coloring fixed before proposals, while adaptive ordinary coloring matches it. The three initializations within each region are repeated measurements of that region.

\FloatBarrier
\section{Exploratory fractional route-search schedule}\label{app:routing-schedule}

\paragraph{Scheduling question.}
When local-search agents share vehicle routes, exclusive route access limits which searches can run concurrently. We compare the duration of ordinary and fractional \emph{preemptible execution schedules} under these conflicts. Both schedules use five agents solving a capacitated vehicle-routing problem (CVRP), with the same search budget for every agent.

\paragraph{Instances, graph, and search.}
We use four public CVRPLIB/DIMACS routing inputs~\cite{cvrplib-dimacs}: Loggi-n401-k23 and Loggi-n501-k24, and ORTEC-n242-k12 and ORTEC-n323-k21. Every schedule for an input starts from the same deterministic feasible Clarke--Wright incumbent. A fixed route-proximity rule selects a five-route cycle, and agent $i$ searches two consecutive routes of that cycle. Adjacent agents share a route and cannot hold their route-pair locks concurrently; nonadjacent agents have disjoint locks. This rule deliberately induces $C_5$ rather than estimating a naturally occurring conflict graph. Each agent performs stateful stochastic two-route local search with a deterministic per-instance, per-repeat seed and exactly 500,000 candidate-move evaluations per schedule.

\paragraph{Matched schedules and measurements.}
The ordinary schedule uses the fixed color classes $\{0,2\}$, $\{1,3\}$, and $\{4\}$, assigning the full 500,000 evaluations to an agent on its one visit. The fractional schedule uses the five independent pairs $\{0,2\}$, $\{1,3\}$, $\{2,4\}$, $\{3,0\}$, and $\{4,1\}$, assigning 250,000 evaluations per visit; every agent appears twice. Both therefore run the same per-agent search budget. The two compatible agents in a wave execute concurrently; waves execute sequentially. We measure elapsed wall time for the entire schedule, summed process CPU time, final route cost, and feasibility. Three paired repeats per input use the same initial incumbent and corresponding seeds, with schedule order reversed in the middle repeat. The nominal critical-path length is $2.5$ rather than $3$ full-budget waves, a $16.7\%$ difference before overhead.

\begin{table}[h]
\caption{Wall-time reduction under fractional route-search scheduling. Each pair gives ordinary $\rightarrow$ fractional completion time in seconds. Reductions are medians over three paired runs. Route-cost differences are fractional minus ordinary; positive values indicate higher route cost.}
\label{tab:routing-schedule-full}
\centering
\small
\begin{tabular}{lcccc}
\toprule
Input & Repeat 1 & Repeat 2 & Repeat 3 & Median reduction \\
\midrule
Loggi-n401-k23 & $7.472\to6.370$ & $7.462\to6.366$ & $7.542\to6.322$ & $14.74\%$ \\
Loggi-n501-k24 & $8.605\to7.249$ & $8.395\to7.289$ & $8.493\to7.196$ & $15.27\%$ \\
ORTEC-n242-k12 & $7.909\to6.543$ & $7.670\to6.633$ & $7.555\to6.561$ & $13.52\%$ \\
ORTEC-n323-k21 & $6.888\to5.894$ & $6.944\to5.850$ & $7.058\to5.916$ & $15.76\%$ \\
\midrule
\multicolumn{5}{l}{Route-cost differences by repeat:}\\
\multicolumn{5}{l}{Loggi-n401-k23: $+56,0,0$; Loggi-n501-k24: $-96,-210,-1$.}\\
\multicolumn{5}{l}{ORTEC-n242-k12: $+50,-118,+6$; ORTEC-n323-k21: $0,0,0$.}\\
\bottomrule
\end{tabular}
\end{table}

\paragraph{Results.}
Table~\ref{tab:routing-schedule-full} shows faster fractional execution in all twelve paired runs. Paired wall-time reductions range from 13.16\% to 17.28\%, with a pooled median of 15.00\%. The fractional/integral ratio of summed agent CPU time ranges from 0.991 to 1.027, with median 1.008, consistent with equal search budgets. Final route cost improves in four pairs, ties in five, and worsens in three.

\FloatBarrier
\section{Allocation and verified integration on controlled tasks}\label{app:synthetic}
Binary optimization problems with specified losses and success probabilities let us examine allocation and convergence under fully known conditions. We use them to compare deadlines, check the rate calculation, and test merge behavior exhaustively. We then revisit the three-agent repair pilot that supplied control task C0 in Section~\ref{sec:experiments}. All time comparisons in this section use the stated cost model.

\paragraph{Allocation and deadline.}
Use $F(x)=\sum_{b=1}^6(1-x_b)$ on six binary blocks and three agents. A successful call repairs its assigned block exactly; otherwise it leaves it unchanged. Let $R=(r_{ib})$ contain rates governing the probability that agent $i$ succeeds on block $b$ before deadline $h$. Define synthetic success probabilities by
\begin{equation}\label{eq:synthetic}
\begin{gathered}
\eta_{ib}(h)\eqdef1-e^{-r_{ib}h},\quad d_{\mathrm{col}}(h)\eqdef h+.4,\\
R\eqdef
\begin{pmatrix}
.90&.90&.03&.03&.03&.03\\
.03&.03&.60&.60&.03&.03\\
.03&.03&.03&.03&.25&.25
\end{pmatrix}.
\end{gathered}
\end{equation}
Fresh success events are independent across calls in this illustration, although Theorem~\ref{thm:rate} does not need that assumption. Uniform allocation assigns each agent to a distinct uniformly sampled block, so $q_{ib}=\nicefrac16$. Optimized allocation solves \eqref{eq:lp} without a binding token constraint. Both policies use the same three agents and the same deadline; in these solutions every agent is assigned in every round. Thus their per-round cost is equal under equal per-call billing at that deadline.

For a fixed allocation, an unrepaired block $b$ survives a round with probability $1-a_b$. Its expected loss satisfies the exact recursion $\Exp{1-x_b^{k+1}}=(1-a_b)\Exp{1-x_b^k}$. Starting from zero and iterating gives
\begin{equation}\label{eq:exact}
\begin{gathered}
\tfrac{\Exp{F(x^K)}}6=\tfrac16\sum_{b=1}^6(1-a_b)^K,\\
\max_{b\in[6]}(1-a_b)^K\overset{\eqref{eq:coverage}}{=}(1-a)^K.
\end{gathered}
\end{equation}
The second quantity is the largest expected remaining loss when initially only one block is unrepaired. It measures exactly the worst-block criterion optimized by the linear program. For that criterion, the number of rounds needed to reach $0.01$ is $\lceil\log(0.01)/\log(1-a)\rceil$.

\begin{table}[h]
\caption{Effect of the call deadline $h$ on uniform and optimized (Opt.) allocation in the six-block example. Coverage $a$ is the smallest per-round repair probability across blocks. Modeled time is the time needed to reduce the worst-block expected remaining loss to $0.01$, using round duration $h+0.4$.}
\label{tab:deadlines}
\centering
\begin{tabular}{rrrrr}
\toprule
&\multicolumn{2}{c}{Coverage $a$}&\multicolumn{2}{c}{Modeled time}\\
$h$ & Uniform & Opt. & Uniform & Opt.\\
\midrule
0.25 & 0.0126 & 0.0340 & 236.6 & 87.1 \\
0.5 & 0.0245 & 0.0662 & 167.4 & 61.2 \\
1 & 0.0467 & 0.1254 & 135.8 & 49.0 \\
2 & 0.0850 & 0.2211 & 124.8 & 45.6 \\
4 & 0.1430 & 0.3459 & 132.0 & 48.4 \\
8 & 0.2152 & 0.4520 & 168.0 & 67.2 \\
\bottomrule
\end{tabular}
\end{table}
Table~\ref{tab:deadlines} compares uniform and optimized allocation at each deadline using color-only costs. On this separable task, the union and color candidate coincide, so deduplicating them gives the same cost for optimistic integration. Each deadline is evaluated separately, without a shared total token budget across deadlines. Longer calls are more likely to succeed but can increase completion time, while very short calls repeatedly incur the fixed overhead. Both allocation policies attain their shortest modeled completion time at $h=2$ on this grid.

At $h=2$, optimized allocation equalizes all six coverage values at approximately $0.2211$, whereas uniform allocation gives $0.1585$, $0.1359$, and $0.0850$ to the three pairs of blocks. The corresponding worst-block horizons are 19 and 52 rounds, or 45.6 and 124.8 modeled time units. 

\paragraph{Checking the rate and merge identities.}
We compare the exact mean-loss formula with 20,000 seeded independent runs at $h=2$ over 20 rounds. The simulated normalized mean loss is $0.006367$, with standard error $0.000232$, compared with the exact value $0.006752$. We also verify the matching decomposition used to realize the allocation.

For a six-block path with objective $\sum_b(1-x_b)+0.2\sum_{b=1}^5(x_b-x_{b+1})^2$, odd and even blocks form two independent color classes. Enumerating all 64 states and all 64 subsets of proposed edits verifies the merge identity. Enumerating the matching and success outcomes at each state also verifies the expected-drift bound. Across the 4,096 state/subset pairs, the union gives a larger gain than the best color candidate in 1,369 cases and ties in 2,727. Optimistic integration satisfies the merge bound throughout.

\paragraph{Interference on a complete graph.}\label{sec:examples}
A concrete instance of Theorem~\ref{thm:tight} illustrates adverse interactions. For $z\in\{0,1\}^m$, start at $z=0$ and let each proposal set one coordinate to one under the objective
\begin{equation}\label{eq:clique}
F(z)\eqdef m-\sum_{b=1}^m z_b+2\sum_{b<c}z_bz_c.
\end{equation}
Each proposal gains one individually. Applying $r$ proposals gives
\begin{equation}\label{eq:clique-gain}
F(0)-F(z)=r-r(r-1)=r(2-r)\leq1,
\end{equation}
where $r\in\{0,\ldots,m\}$. Thus the best subset gains one out of a total singleton gain of $m$, attaining the factor $1/m$ for the complete graph.
 For $m=2,\ldots,8$, applying all edits gives gain $m(2-m)$, while selecting one color class gives gain one. Optimistic integration therefore selects the color candidate, rejecting the worsening union for $m\geq3$.

\paragraph{The three-agent repair pilot.}
Three AI agents independently repaired a parser, discounted-total calculator, and receipt formatter from one common faulty snapshot. Before the original calls we fixed 24 checks: 16 single-function checks and eight adjacent-stage checks. The block-dependency graph is the three-block path, with the two endpoint functions sharing a color. Agents saw the task specification and initial code, not the central checks or other proposals. We replayed the original frozen proposals and made three new calls to the same agents, which retained their earlier task context. The second batch therefore checks repeat behavior with that context available. Both batches gave the same results in Table~\ref{tab:pilot}.

\begin{table}[h]
\caption{Repair progress and verification cost in the three-agent pipeline pilot. Arrows show the number of failing checks from the initial artifact through successive merges; ``Evaluations'' counts charged objective evaluations. Within each batch, all policies reuse the same three frozen proposals. The original and repeat batches give the same counts.}
\label{tab:pilot}
\centering
\begin{tabular}{lcc}
\toprule
Merge policy & Failed checks & Evaluations\\
\midrule
Color only & $23\to12\to0$ & 5\\
Optimistic & $23\to0$ & 5\\
Merge all & $23\to0$ & 2\\
\bottomrule
\end{tabular}
\end{table}
The initial singleton gains were $7,4,4$. The best color class gained 11, while the union gained 23, giving an interaction surplus of eight. After the first color-class merge, re-evaluating the stored middle-block proposal raised its gain to 12 and completed the repair without another agent call. Optimistic integration completed repair in one merge, but both verified policies used five objective evaluations. Merge-all used two.

These counts include the initial objective and evaluations of changed singleton candidates; optimistic integration also evaluates the union. Color-candidate loss follows from additivity, and evaluations used only to check that identity are excluded. Thus the pilot illustrates how favorable interactions can save a merge step without reducing verification cost.

For the original proposals, enumerating the eight programs formed by choosing the original or repaired version of each block gives a zero optimum and block-error-bound constant $\mu=15/23$ on that finite domain. Both batches used the same parent model with different module assignments, and the repeat calls retained the agents' earlier context. The benchmark in Section~\ref{sec:experiments} reuses this task with six additional hidden checks and starts every call in a fresh context.

\FloatBarrier
\section{Proofs}\label{app:proofs}
\subsection{Proof of Theorem~\ref{thm:rate}}
Theorem~\ref{thm:rate} combines three inequalities: the merge retains local gains, conditional agent-quality bounds relate those gains to exact block improvement, and the block error bound relates that improvement to global error. 
\begin{proof}
\textbf{Expected local progress.}
Write $\mathbb E_k[\cdot]\eqdef\mathbb E[\cdot\mid\cH_k]$ and $\mathbb E_{k,\cM}[\cdot]\eqdef\mathbb E[\cdot\mid\cH_k,\cM]$. Condition on the history before drawing the matching. Apply Proposition~\ref{prop:optimistic} (or Theorem~\ref{thm:fractional-merge} for the safe-only method), then condition on each matching and use Assumption~\ref{ass:oracle}:
\begin{equation}\label{eq:drift}
\begin{aligned}
\mathbb E_k[F(x^k)-F(x^{k+1})]
&\overset{\eqref{eq:optimistic-bound}}{\geq}
\frac1\Lambda\sum_{\cM}p_{\cM}\sum_{(i,b)\in\cM}\mathbb E_{k,\cM}[g_{ib}^k]\\
&\overset{\eqref{eq:oracle}}{\geq}
\frac1\Lambda\sum_{\cM}p_{\cM}\sum_{(i,b)\in\cM}\eta_{ib}(h)D_b(x^k)\\
&\overset{\eqref{eq:coverage}}{=}
\frac1\Lambda\sum_b a_bD_b(x^k)
\geq\frac a\Lambda\sum_bD_b(x^k)
\overset{\eqref{eq:error},\eqref{eq:coverage}}{\geq}\theta_\lambda\Delta(x^k).
\end{aligned}
\end{equation}
The penultimate inequality uses $a_b\geq a$ from~\eqref{eq:coverage} and $D_b(x^k)\geq0$, which follows because leaving the block unchanged is feasible in \eqref{eq:task}.

\textbf{From progress to error.}
Subtract the improvement from the old objective gap:
\begin{equation}\label{eq:recursion}
\begin{aligned}
\mathbb E_k[\Delta(x^{k+1})]
&\overset{\eqref{eq:task}}{=}\Delta(x^k)-\mathbb E_k[F(x^k)-F(x^{k+1})]\\
&\overset{\eqref{eq:drift}}{\leq}(1-\theta_\lambda)\Delta(x^k).
\end{aligned}
\end{equation}
Write $e_k\eqdef\Exp{\Delta(x^k)}$. Taking expectations in the conditional bound gives
\begin{equation}\label{eq:mean-recursion}
e_{k+1}\overset{\eqref{eq:recursion}}{\leq}(1-\theta_\lambda)e_k.
\end{equation}
The base case is $e_0=\Delta(x^0)$. If $e_k\leq(1-\theta_\lambda)^ke_0$, then
\begin{equation}\label{eq:induction}
e_{k+1}\overset{\eqref{eq:mean-recursion}}{\leq}(1-\theta_\lambda)e_k
\leq(1-\theta_\lambda)^{k+1}e_0,
\end{equation}
where the second inequality is the induction hypothesis and uses $1-\theta_\lambda\geq0$. This proves \eqref{eq:rate}. If $\theta_\lambda=0$, the conclusion gives monotonicity only. If $\theta_\lambda=1$, nonnegativity forces a zero gap after one round; the case $K=0$ means the unchanged initial gap.
\end{proof}
\FloatBarrier

\section{Extended related work}\label{app:related}

\paragraph{Graph-theoretic context.}
Broader work on graph coloring and related extremal questions includes \citep{hilton19735,demidovich2023reconstruction,demidovich2023tight,demidovich2023cycle,sagdeev2019frankl}.

\paragraph{Agent systems.}
Further work on agent interaction and orchestration includes \citep{wu2024,chen2024,du2024,dang2025}; empirical studies of collaboration include \citep{cemri2025,kim2026,park2026}.
\citet{ilin2026localworkglobalprogress} analyze convergence of corrected local search with verified merging.

\paragraph{Distributed optimization and training.}
Work on time, allocation, and communication includes \citep{tyurin2023,maranjyan2025ata,maranjyan2026,fradin2025}.
Related work on distributed model training includes \citep{goyal2018accuratelargeminibatchsgd,konecny2017federatedlearningstrategiesimproving,NEURIPS2024_e0b6f389,li2024convergencefedproxextrapolationinexact,mishchenko2022proxskip,panferov2024correlated,demidovich2025mast,demidovich2025methods,sadiev2022federatedoptimizationalgorithmsrandom,sadiev2025improved}; other optimization and model-adaptation methods include \citep{khirirat2025better,sokolov2025bernoulli,ilin2025thanos}.

\end{document}